\documentclass{article}

\RequirePackage[OT1]{fontenc} \RequirePackage{amsthm,amsmath} \RequirePackage[authoryear]{natbib}

\usepackage{amsfonts, amssymb, graphicx, mathtools, epstopdf} 
\usepackage{accents, appendix, xcolor}
\usepackage{enumitem}

\usepackage[colorlinks, citecolor=blue, urlcolor=blue]{hyperref}

\usepackage{caption}
\usepackage{subcaption, multirow, bigfoot}

\usepackage{multibib}

\newcites{appendix}{References}

\theoremstyle{definition}

\theoremstyle{plain}
\newtheorem{theorem}{Theorem}
\newtheorem{assumption}{Assumption}

\newtheorem{lemma}{Lemma} 

\theoremstyle{remark}
\newtheorem*{remark}{Remark}

\DeclareMathOperator*{\argmax}{arg\,max}
\newcommand{\stdt}{{\sqrt{\operatorname{Var} \tau}}}
\newcommand{\Var}{{\operatorname{Var}\, }}
\newcommand{\alt}{{\textup{alt}}}
\newcommand{\E}{{\mathbb{E}}}
\newcommand{\diff}{{\textup{diff}}}

\title{\normalsize Estimating the prevalance of indirect effects and other spillovers}
\author{\normalsize David Choi}
\date{\normalsize \today}

\begin{document}
\maketitle

\abstract{
In settings where interference between units is possible, we define the prevalence of indirect effects to be the number of units who are affected by the treatment of others. This quantity does not fully identify an indirect effect, but may be used to show whether such effects are widely prevalent. Given a randomized experiment with binary-valued outcomes, methods are presented for conservative point estimation and one-sided interval estimation. No assumptions beyond randomization of treatment are required, allowing for usage in settings where models or assumptions on interference might be questionable. To show asymptotic coverage of our intervals in settings not covered by existing results, we provide a central limit theorem that combines local dependence and sampling without replacement. Consistency and minimax properties of the point estimator are shown as well. The approach is demonstrated on an experiment in which students were treated for a highly transmissible parasitic infection, for which we find that a significant fraction of students were affected by the treatment of schools other than their own. 
 }


\section{Introduction}

In a study where units may be affected by others (``interference between units''), it may be of interest to assign a measure of strength to the underlying interference. One possible approach is to estimate the indirect treatment effect -- i.e, the effect of treating one unit on the outcome of others \citep{tchetgen2012causal, li2022network}. However, confidence intervals for such effects require assumptions or bounds on interference (that is, on the extent to which units can affect each other), which may be unsuitable for poorly understood settings. If no such assumptions can be made, a very incomplete answer can be found by testing the null hypothesis of no interference \citep{athey2017exact, pouget2019testing}. However, rejection of the null only implies that at least one unit is affected by the treatment of others, leaving open the possibility that nearly all units are affected only by their own treatment.

In this paper, we show that it is possible to invert a test of no interference to find a one-sided confidence interval for the number of units who are affected by the treatment of others. This does not identify a indirect effect; for example, it does not answer whether such effects are positive or negative. However, it can be used to show that indirect effects are widely prevalent, without reliance on bounds or assumptions on interference. We also propose a conservative point estimator that is consistent and minimax for a lower bound on the quantity of interest. The approach requires a randomized experiment with known design and binary-valued outcomes. No formal assumptions on interference are required. As a result, our estimates may safely use crude network proxies such as geographic proximity. By ``safely'', we mean that failure of the elicited network to accurately reflect the true generative mechanism may result in reduced power or increased uncertainty, but not loss of coverage or anti-conservative bias.

Existing central limit theorems will not generally apply to our test statistic -- due to its networked nature -- except in settings where treatments are assigned independently for each unit. To extend the range of experiment designs for which asymptotic normality can be shown, we provide a new central limit theorem that considers a sum of $N$ variables whose dependency graph would be sparse if treatments were independently randomized, but now exhibit global dependence due to the assignment of treatment by sampling without replacement. 





We demonstrate the proposed approach on an experiment described in \cite{miguel2004worms}, in which free treatment of intestinal parasites was provided to a quasi-random selection of schools in a region where infection rates were high. Treated students were susceptible to reinfection by untreated peers, resulting in interference between units. Our point estimate indicates that for at least 9.6\% of the students in the experiment (one-sided 95\% CI: at least 0.3\%), whether or not they experienced any parasitic infections in the following year depended not only on whether their own school received the deworming treatment, but on the treatment of other schools as well. 

\section{Related Literature}

Indirect effects are discussed for vaccination in \cite{tchetgen2012causal} assuming non-interfering neighborhoods, and for spatial-temporal settings in \cite{wang2021causal} assuming bounds on interference, though not on its exact structure. In both papers, indirect effects are localized to the effect on one's neighboring units, while \cite{hu2021average} and \cite{li2022random} consider a more global notion whose inference requires stronger assumptions; for example, interval estimation of the indirect effect of \cite{li2022random} requires anonymous interference under an observed and correctly specified network generated by a low rank graphon. In all of these papers, the effects are defined in expectation over the random treatment assignment, analogous to the EATE (or direct effect) discussed in \cite{savje2017average}.

Testing for the null hypothesis of no interference is studied in \cite{athey2017exact}, \cite{aronow2012general}, \cite{pouget2019testing}, \cite{basse2019randomization}, and \cite{han2022detecting}. \cite{pouget2019testing} considers testing for cluster randomized experiments. \cite{han2022detecting} considers sequential settings in which treatment is gradually released to increasing numbers of units. Other recent works on interference include cluster randomized experiments \citep{leung2022rate, park2023assumption}; peer effects \citep{li2019randomization,basse2019randomization_peer}; observational studies \citep{ogburn2022causal,forastiere2020identification}; and others \citep{li2022network,tchetgen2021auto,cortez2022staggered}. 

The recent paper \citep{choi2021randomization} considered questions of the form ``relatively speaking, were the units that received treatment affected more positively than those that received the control?'' and also gave lower bounds on the number of units that were affected by any treatment (including their own). We will use a convex relaxation from this paper; however, we also go beyond their work in significant ways. In contrast to \cite{choi2021randomization}, we bound the number of units that were affected by the treatment of others, as opposed to their own treatment. To do so, we compare the observed outcomes to a new counterfactual that we term {\it isolated treatment}, to which their approach does not apply. In particular, inference for our estimands requires a more indirect approach in which we invert a test statistic, as opposed to their work in which the error of the point estimator can be bounded directly. Additionally, we note that our central limit theorem fills a gap in \cite{choi2021randomization} -- and also in other earlier works that assumed independent Bernoulli treatment randomization, such as \cite{aronow2017estimating} --  by applying to experiment designs not covered by the dependency graph CLT used in these papers.


\section{Problem Formulation}

\subsection{Preliminaries} \label{sec: preliminaries}

Given a randomized experiment on $N$ units, let $X = (X_1,\ldots,X_N)$ and $Y = (Y_1,\ldots,Y_N)$ denote the binary-valued treatments and outcomes of each unit. We assume that the experiment design and hence the distribution of $X$ is known. As as we place no restrictions on interference, each unit's outcome may potentially depend on all $N$ treatments. As a result, for each unit $i \in [N]$ it holds that
\begin{align} \label{eq: f}
Y_i & = f_i(X_1,\ldots,X_N) 
\end{align}
for some set of potential outcome mappings $\{f_i\}_{i=1}^N$. 


Given an observed network for the $N$ units, let $G \in \{0,1\}^{N \times N}$ denote its adjacency matrix, which has no self-connections so that $G_{ii}=0$ for all $i \in [N]$. Given $G$, let $W = (W_1,\ldots,W_N) \in \{0,1\}^N$ denote a binary summary statistic of the treatment received by each unit's neighbors. For example, $W_i$ could indicate whether unit $i$'s number of treated neighbors exceeds a threshold, so that $W_i$ is given by
\[ W_i = \begin{cases} 1 & \text{ if }\sum_{j:G_{ij}=1} X_j \geq \gamma_i \\ 0 & \text{ otherwise} \end{cases}\]
for some vector of unit-specific thresholds $\gamma = (\gamma_1,\ldots,\gamma_N)$. We will refer to $X_i$ and $W_i$ respectively as the direct and indirect treatment of unit $i$.

\subsection{Estimand}

Let $\theta^* = (\theta_1^*,\ldots,\theta_N^*)$ equal the outcome of each unit under a counterfactual of isolated treatment -- that is, let $\theta_i^*$ equal the outcome of unit $i$ if they receive their assigned treatment $X_i$, while all other units receive the control:
\begin{align*}
\theta_i^* & = f_i(0,\ldots,0,\underbrace{X_i}_{\textup{$i$th position}},0,\ldots,0).
\end{align*}
Our estimand $\psi$ will be the absolute difference between $Y$ and $\theta^*$,
\[ \psi = \sum_{i=1}^N |Y_i - \theta_i^*|. \] 
For binary outcomes, this quantity can be interpreted as the number of units who were affected by the treatment of another unit, in the sense that their outcome differed from what would have occurred had they received their treatment in isolation. Alternately, it may be viewed as a lower bound on the number of units whose outcome mapping $f_i$ is non-constant in the treatment of others.

Our estimand $\psi$ may be generalized, by letting the counterfactual vector $\theta^*$ be any binary vector satisfying
\[ \theta_i^* = \tilde{f}_i(X_i), \qquad \forall\ i \in [N], \]
for some choice of functions $\{\tilde{f}_i\}_{i=1}^N$. For example, suppose that units are divided into groups (such as schools or villages) and treatment is randomized at the group level, so that units in the same group receive the same treatment. In such cases we may define $\tilde{f}_i$ to equal the outcome of unit $i$ if their group receives treatment $X_i$, while all other groups receive the control. Then the condition $Y_i \neq \theta_i$ implies that unit $i$'s outcome differed from what would have occurred had all other groups received the control, and the estimand $\psi$ may be interpreted as the number of units who were affected by the treatment of other groups (or as a lower bound on the number of units whose outcome mapping $f_i$ is non-constant in the treatment of other groups).

\subsection{Interval Estimation}

To find a one-sided interval for $\psi$, our approach will be to define a test statistic $\tau(X, \theta^*)$ that can be used to test and reject hypotheses for $\theta^*$, the unknown vector of of counterfactual outcomes. In principle, a confidence set for $\theta^*$ is given by the set of all non-rejected hypotheses, and the image of $\psi$ over this set induces a confidence set for $\psi$. 

\subsubsection{Test Statistic}

To test a hypothesis for $\theta^*$, our test statistic $\tau$ will compare the values of the hypothetical $\theta^*$ for units with low and high levels of indirect treatment, as measured by the indicator $W$. 
For example, we could let $\tau(X,\theta^*)$ equal the inverse propensity weighted comparison
\begin{equation} \label{eq: tau}
\tau(X, \theta^*) = \sum_{i=1}^N \left( \frac{W_i \theta_i^*}{P(W_i=1|X_i)} - \frac{(1-W_i) \theta_i^*}{P(W_i=0|X_i)}\right),
\end{equation}
where $P(W_i=1|X_i)$ conditions on the observed value of $X_i$. It can be seen that $\tau(X, \theta^*) = u^T\theta^*$ for $u \in \mathbb{R}^N$ given by 
\begin{align} \label{eq: u}
 u_i & = \frac{W_i}{P(W_i=1|X_i)} - \frac{1-W_i}{P(W_i=0|X_i)} & i \in [N],
\end{align}
and that $\E[u_i|X_i] = 0$ for all $i\in [N]$, implying that $\E\tau = 0$. 

Let $V(X,\theta^*)$ denote an estimator of $\Var \tau$ given by
\begin{align} \label{eq: V}
V(X, \theta^*) = \sum_{i=1}^N \sum_{j=1}^N \theta_i^* \theta_j^* \E[u_iu_j| X_i,X_j].
\end{align}
As the experimental design and the functional form of $u$ is known, to evaluate \eqref{eq: V} we can compute or simulate the conditional expectations $\E[u_iu_j | X_i, X_j]$ for $i,j \in [N]$.

\subsubsection{Interval Construction} Under certain conditions (see Section \ref{sec: theory}), it can be shown that $\tau / \sqrt{\Var \tau}$ is asymptotically normal and that $V = (\Var \tau)(1+o_P(1))$, in which case it holds that
\[ \frac{\tau(X, \theta^*)}{\sqrt[\uproot{5}+]{V(X, \theta^*)}} \rightarrow_d N(0,1),\]
where $\sqrt[\uproot{5}+]{\phantom{.}}$ denotes the function 
\begin{align} \label{eq: sqrt pos}
 \sqrt[\uproot{5}+]{x} &= \begin{cases} \sqrt{x} & \textup{ if } x \geq 0 \\ -\infty & \textup{ otherwise, }\end{cases}
 \end{align}
which extends square root to $\mathbb{R}$ while preserving concavity.

In such cases, with probability converging to $(1 - 2\alpha)$ it holds that
\[  |\tau(X,\theta^*)| \leq z_{1-\alpha} \sqrt[\uproot{5}+]{V(X, \theta^*)} \]
where $z_{1-\alpha}$ denotes the $(1-\alpha)$ quantile of a standard normal. It follows that an asymptotic $(1-2\alpha)$ level confidence set $S$ for the unknown $\theta^*$ is given by 
\begin{align} \label{eq: S}
S = \bigg\{\theta \in \{0,1\}^{N}: |\tau(X,\theta)| \leq z_{1-\alpha} \sqrt[\uproot{5}+]{V(X,\theta)}\bigg\},
\end{align}
and hence an asymptotic $(1-2\alpha)$ level confidence lower bound (or one-sided interval) on $\psi$ is given by its smallest value over all $\theta \in S$:
\begin{equation}\label{eq: lower bound}
 \min_{\theta \in \{0,1\}^{N}} \sum_{i=1}^N |Y_i - \theta_i|\quad  \textup{ subject to } \quad |\tau(X,\theta)| \leq z_{1-\alpha} \sqrt[\uproot{5}+]{V(X,\theta)}
\end{equation}

\subsubsection{Computation}
To compute the one-sided interval \eqref{eq: lower bound}, we observe that $V(X, \theta) = \theta^T Q\theta$ for $Q \in \mathbb{R}^{N \times N}$ given by
\begin{align*}
Q_{ij} & = \E[u_i u_j | X_i, X_j]
\end{align*}
and hence we may substitute $\tau = u^T\theta$ and $V = \theta^T Q \theta$ into \eqref{eq: lower bound} to yield
\begin{equation}\label{eq: optimization}
 \min_{\theta \in \{0,1\}^{N}} \sum_{i=1}^N |Y_i - \theta_i|\quad  \textup{ subject to } \quad |u^T\theta| \leq z_{1-\alpha}\sqrt[\uproot{5}+]{\theta^T Q \theta}
\end{equation}
Note that \eqref{eq: optimization} must be globally solved (or lower bounded) in order for our interval to have provable coverage guarantees. 

To proceed, we apply a convex relaxation that can be found in \cite{choi2021randomization}. First, we find a diagonal matrix $D$ such that $Q - D$ is negative semidefinite. This can be done by solving the semidefinite program
\begin{align} \label{eq: sdp}
    \min_{D \in \mathbb{R}^{N\times N}} \sum_{i=1}^N D_{ii} \quad \textup{subject to } D \textup{ diagonal, $D \geq 0,$ and } Q-D \preceq 0 
\end{align}
where $\geq 0$ denotes nonnegativity and $\preceq 0$ denotes negative semidefinite. Using the fact that $x = x^2$ for $x\in \{0,1\}$ and hence that
\[ \theta^TQ \theta = \theta^T(Q - D)\theta + \sum_{i=1}^N D_{ii}\theta_i \quad \text{ for all } \theta \in \{0,1\}^N, \]
we can rewrite \eqref{eq: optimization} equivalently as
\begin{equation}\label{eq: optimization 2}
 \min_{\theta \in \{0,1\}^{N}} \sum_{i=1}^N |Y_i - \theta_i|\quad  \textup{ subject to } \quad |u^T\theta| \leq z_{1-\alpha} \sqrt[\uproot{15}+]{\theta^T (Q-D) \theta + \sum_i D_{ii}\theta_i}
\end{equation}
Since $Q - D \preceq 0$, the constraint 
\[|u^T\theta| \leq z_{1-\alpha} \sqrt[\uproot{15}+]{\theta^T (Q-D) \theta + \sum_i D_{ii}\theta_i}\]
is convex is $\theta$. As a result, relaxing the binary constraint of \eqref{eq: optimization 2} to $\theta \in [0,1]^{N}$ results in a convex optimization problem whose global optimum can be efficiently found. Doing so returns a lower bound on \eqref{eq: optimization 2}, or equivalently a conservative one-sided prediction interval for $\psi$.

\subsection{Conservative Point Estimate} 
Given a test statistic whose mean is zero and which takes form $u^T\theta^*$ for some random vector $u \in \mathbb{R}^N$, let $L$ denote the quantity
\begin{align} \label{eq: L}
L= \frac{|u^T(Y - \theta^*)|}{\|u\|_\infty} 
\end{align}
By H\"older's inequality, it follows that $\psi$ is lower bounded by $L$
 \begin{align} \label{eq: holder}
\sum_{i=1}^N |Y_i - \theta_i^*| \geq \frac{|u^T(Y - \theta^*)|}{\|u\|_\infty}  
\end{align}
We propose $\hat{L}$ as a point estimate of this lower bound:
\begin{align} \label{eq: hat L}
 \hat{L}  = \frac{|u^T(Y - 0.5)|}{ \|u\|_\infty} 
\end{align}
In Section \ref{sec: theory}, conditions are given under which $\hat{L}$ is consistent and minimax for the lower bound $L$.

\begin{remark} 
It can be seen that the lower bound \eqref{eq: holder} can be particularly loose if $u$ has a small number of entries with substantially larger magnitude than the others. For $\tau$ given by \eqref{eq: tau}, this can occur if $P(W_i=1|X_i)$ or $P(W_i=0|X_i)$ is near zero for a small set of units. In such settings, for point estimation it may be preferable to use the statistic $\tau^{\alt}$ given by
\begin{align} \label{eq: tau alt}
\tau^{\alt}(X,\theta^*) = \sum_{i:W_i=1} \pi_{i1} \theta_i^* - \sum_{i:W_i=0}\pi_{i0} \theta_i^*,
\end{align}
where the weights $\pi_{i1}$ and $\pi_{i0}$ are given by
\begin{align*}
\pi_{i1} &= \min\left(\frac{P(W_i=0|X_i)}{P(W_i=1|X_i)}, 1\right) &
\pi_{i0} &= \min\left(\frac{P(W_i=1|X_i)}{P(W_i=0|X_i)}, 1\right)
\end{align*}
It can be seen that $\tau^{\alt} = u^T \theta^*$ for $u$ given by
\begin{align}\label{eq: u alt}
 u_i = W_i \pi_{i1} - (1-W_i)\pi_{i0}, \qquad i \in [N].
\end{align}
When the probability of indirect treatment is particularly unbalanced for a small set of units, it can be seen that $u$ given by \eqref{eq: u alt} may weigh units more evenly than \eqref{eq: u}. As a result, the lower bound $L$ may be larger (and hence tighter) if $\tau^{\alt}$ rather than $\tau$ is used for point estimation. (Such decisions should be made before the outcomes are examined, to avoid specification search.)
\end{remark}

\begin{remark}
 It can be seen from \eqref{eq: holder} that the lower bound $L$ will be exactly tight for $\psi$ if the following conditions both hold:
\begin{enumerate}
	\item $u_i(Y_i - \theta_i) \geq 0$ for all $i \in [N]$ -- that is, when $Y_i - \theta_i \neq 0$, its sign is predicted by $u_i$.
	\item $|u_i|$ is constant over $i \in [N]$
\end{enumerate}
As the conditions listed above may be difficult to exactly achieve, we generally expect $L$ to be a conservative proxy for $\psi$, as might be expected in the absence of a structural model or other assumptions on interference. Nonetheless, the conditions may give intuition as to the ideal experiment design. If $\tau$ given by \eqref{eq: tau} or $\tau^{\alt}$ given by \eqref{eq: tau alt} is used as the test statistic, the first condition suggests that the indirect treatment $W$ should be highly aligned with $Y - \theta$, as might occur when a strong causal relationship exists between the two. The second condition suggests that $P(W_i=1|X_i)$ should be close to $0.5$ for all units, so that each observation has equal leverage. 

\end{remark}

\section{Asymptotics} \label{sec: theory}

We give results on asymptotic normality of $\tau$, coverage of interval estimates, and consistent minimax point estimation of the lower bound $L$.

\subsection{Assumptions}

Assumption \ref{as: random sampling} describes a simple choice of experiment design in which the treatments are assigned by sampling without replacement.

\begin{assumption} \label{as: random sampling}
The unit treatments $X_1,\ldots,X_N$ are generated by sampling $T$ units without replacement, and there exists $\rho \in (0,1)$ such that $\rho \leq T/N  \leq 1-\rho$.
\end{assumption}

Assumption \ref{as: degree} restricts the adjacency matrix $G$ to have bounded degree. Note that this assumption does not limit interference between units, as it places no restrictions on the outcome mappings $\{f_i\}_{i=1}^N$.

\begin{assumption} \label{as: degree}
The adjacency matrix $G$ has in-degree and out-degree bounded by $d_{\max}$.
\end{assumption}

Assumption \ref{as: tau} places conditions on the statistic $\tau$ and on the counterfactual of interest $\theta^*$. 

\begin{assumption} \label{as: tau}
The statistic $\tau$ can be written as
$\tau(X, \theta^*) = u^T\theta^*$ where $u$  and $\theta^*$ are functions of $X$ that satisfy for some $C > 0$ that
\begin{enumerate}[label=\text{(A\arabic*)}]
    \item $\mathbb{E}[u_i|X_i = x] = 0$ for all $i \in [N]$ and $x \in \{0,1\}$ \label{item: zero mean}
    \item $|u_i(X)| \leq C$ for all $i \in [N]$ and $X \in \{0,1\}^N$ \label{item: u bounded}
\item For all $i \in [N]$ and $X, X' \in \{0,1\}^N$ it holds that \label{item: local}
\begin{align*}
u_i(X) & = u_i(X') & \textup{if } X_j=X_j' \textup{ for all } j \in \{i\} \cup \{k: G_{ik} = 1\}
\end{align*}
so that $u_i(X)$ depends only on $X_i$ and the entries of $X$ corresponding to $i$'s neighbors in $G$
\item There exist mappings $\tilde{f}_i: \{0,1\}\mapsto \{0,1\}$ such that for all $i \in [N]$ it holds that \label{item: theta}
\[ \theta_i^* = \tilde{f}_i(X_i) \]
\end{enumerate}
\end{assumption}

This assumption requires that the statistic $\tau$ is of the form $u^T\theta^*$, in which the entries of $u$ are required by conditions \ref{item: zero mean}-\ref{item: local} to satisfy a zero-mean condition, to be bounded, and to depend on subvectors of $X$ given by the neighborhood struture of the network $G$; and condition \ref{item: theta} requires each entry of $\theta^*$ to depend only on the corresponding unit's own treatment. We note that \ref{item: zero mean} and \ref{item: local} hold always for $\tau$ and $\tau^{\alt}$ given by \eqref{eq: tau} and \eqref{eq: tau alt}, while \ref{item: u bounded} holds provided that $P(W_i=1|X_i=x)$ is bounded away from 0 and 1 for all $i \in [N]$ and $x \in \{0,1\}$.

While Assumption \ref{as: tau} is required for our results on consistency and interval coverage, it goes beyond the requirements for asymptotic normality. For this reason, our central limit theorem will require a weaker and more generic condition, which we describe below.

\begin{assumption} \label{as: tau clt}
The statistic $\tau$ can be written as $\tau = \sum_{i=1}^N v_i(X)$ where the functions $v_i$ for $i \in [N]$ are bounded in magnitude by
\begin{align} \label{eq: v bounded}
 |v_i(X)| & \leq C, & \forall \ X \in \{0,1\}^N
 \end{align}
 for some $C>0$, and satisfy for all $X$ and $X'$ in $\{0,1\}^N$ that 
\begin{align} \label{eq: v local}
v_i(X) & = v_i(X') & \textup{if } X_j=X_j' \textup{ for all } j \in \{i\} \cup \{k: G_{ik} = 1\}
\end{align}
so that $v_i(X)$ depends only on $X_i$ and the entries of $X$ corresponding to $i$'s neighbors in $G$. 
\end{assumption}

Assumption \ref{as: tau clt} requires only that $\tau$ is a sum of terms that are bounded and depend on subvectors of $X$ given by the neighborhood structure of $G$. In particular, $\tau$ need not be of the form $u^T\theta^*$. It can be seen that Assumption \ref{as: tau} implies Assumption \ref{as: tau clt} for  $v_i(X) = u_i \theta_i^*$.

\subsection{Asymptotic Normality} Theorem \ref{th: clt} gives a central limit theorem under Assumptions \ref{as: random sampling}, \ref{as: degree}, and \ref{as: tau clt}. Due to the global dependence that is induced by sampling without replacement, the theorem is not a straightforward application of local dependence results such as \cite{rinott1994normal}, \cite[Th. 2.7]{chen2004normal}, \cite{ogburn2022causal}, or \cite{chin2018central} (which itself uses \cite{chatterjee2008new}), and may be of independent interest. The theorem is proven in Appendix \ref{sec: CLT proof}. 

\begin{theorem} \label{th: clt}
Assume a sequence of experiments for which Assumptions  \ref{as: random sampling}, \ref{as: degree}, and \ref{as: tau clt} hold, in which $\rho$, $d_{\max}$, and $C$ are constant as $N \rightarrow \infty$, and $\operatorname{Var} (\tau) \geq N^{1/2+\delta}$ for some $\delta > 0$. Then  $(\tau - \mathbb{E}\tau)/\sqrt{\operatorname{Var} \tau} \rightarrow_d N(0,1)$.
\end{theorem}

\begin{remark}
Theorem \ref{th: clt} requires the variance of $\tau$ to exceed $N^{1/2+\delta}$ for some $\delta > 0$, in order to rule out denegerate distributions. This condition may fail to hold if $\theta^*$ has only a vanishing fraction of non-zero entries. Otherwise, we would generally expect the variance of $\tau$ to be of order $O(N)$. (To see this, consider that if the entries of $X$ were indpendently assigned, then $\tau$ would be an unnormalized sum of $N$ locally dependent terms.)
\end{remark}

\begin{remark}
The idea of the proof is to apply a martingale central limit theorem to a Doob martingale for the test statistic. This requires convergence of the conditional variances of the martingale, which can be shown by applying a martingale version of the Azuma-Hoeffding inequality. To show the bounded difference condition required by Azuma-Hoeffding, counting and coupling arguments are used.
\end{remark}

\begin{remark}
If Assumption \ref{as: random sampling} is replaced by an assumption that the treatments are independent Bernoulli variates (which need not be identically distributed), then under Assumptions \ref{as: degree} and \ref{as: tau} the test statistic can be seen to have a dependency graph with degree bounded by $O(d_{\max}^2)$. As a result, a local dependence CLT such as \cite{rinott1994normal} or \cite[Th. 2.7]{chen2004normal} may be straightforwardly applied to show asymptotic normality in such settings.
\end{remark}

\subsection{Interval Coverage}
When the variance requirement of Theorem \ref{th: clt} is not met, the normal-based interval given by \eqref{eq: lower bound} may not have the desired coverage probability. To address this concern, a technical modification to \eqref{eq: lower bound} can be made to remove the variance requirement. The modification is to return the minimum of \eqref{eq: lower bound} and the following optimization problem
\begin{align} \label{eq: backup optimization}
 \min_{\theta \in \{0,1\}^{N}} \sum_{i=1}^N |Y_i - \theta_i|\quad  \textup{ subject to } \quad |u^T\theta| \leq N^{1/3}
\end{align}
The intuition for \eqref{eq: backup optimization} is analogous to that of \eqref{eq: lower bound}, except that Chebychev's inequality is used to bound the quantiles of $\tau$ assuming that $\Var \tau$ fails to exceed $N^{1/2+\delta}$ for any $\delta>0$. As this is precisely when Theorem \ref{th: clt} fails to hold, the required conditions for valid coverage will always be met for at least one of \eqref{eq: lower bound} and \eqref{eq: backup optimization} for each experiment in a sequence of experiments where $N \rightarrow \infty$, so that returning the minimum of the two will ensure asymptotic coverage at the desired probability. 

Theorem \ref{th: coverage} establishes asymptotic coverage for the one-sided interval given by the minimum of \eqref{eq: lower bound} and \eqref{eq: backup optimization}, with no conditions required on $\operatorname{Var} \tau$. It is proven in Appendix \ref{sec: coverage proof}.

\begin{theorem} \label{th: coverage}
Assume a sequence of experiments for which Assumptions \ref{as: random sampling}, \ref{as: degree}, and \ref{as: tau} hold, in which $\rho$, $d_{\max}$, and $C$ are constant as $N \rightarrow \infty$. Then with probability converging to at least $1-2\alpha$, it holds that $\psi$ is lower bounded by the minimum of \eqref{eq: lower bound} and \eqref{eq: backup optimization}.
\end{theorem}

\begin{remark}
In practice, we view \eqref{eq: backup optimization} primarily as a technical device that is needed only to remove nondegeneracy conditions. This is because the constraint in \eqref{eq: backup optimization} that $|u^T\theta| = o(N)$ prevents the objective $\|Y - \theta\|_1$ from being small, except in cases where $u^TY = o(N)$ as well. Thus with the exception of such settings, we expect that \eqref{eq: backup optimization}  will typically be larger than \eqref{eq: lower bound}.
\end{remark}

\begin{remark}
The major difficulty of proving Theorem \ref{th: coverage} is showing consistency of the variance estimator $V$. This is done using the Azuma-Hoeffding inequality, resorting to counting arguments to establish the required bounded difference condition.
\end{remark}

\subsection{Point Estimation} Recall that $\hat{L}$ and $L$ are respectively given by \eqref{eq: hat L} and \eqref{eq: L}, and that $L$ is a lower bound on the estimand $\psi$. Theorem \ref{th: consistency} gives conditions under which the point estimate $\hat{L}$ is consistent for $L$, with $O_P(N^{1/2})$ error. It is proven in Appendix \ref{sec: consistency proof}.

\begin{theorem} \label{th: consistency}
Let the conditions of Theorem \ref{th: coverage} hold, and additionally let 
\begin{align} \label{eq: consistency condition}
|u_i(X)| &\geq C_{\min}, &  \forall\ i \in [N],\ X \in \{0,1\}^N,
\end{align} 
where $C_{\min}$ is constant as $N \rightarrow \infty$. Then
\begin{align}\label{eq: consistency}
|L - \hat{L}| = O_P(N^{1/2})
\end{align}
\end{theorem}

 We note that the requirement \eqref{eq: consistency condition} holds for $\tau$ and $\tau^{\alt}$ given by \eqref{eq: tau} and \eqref{eq: tau alt}, provided that the propensities $P(W_i = 1|X_i=x)$ are bounded away from 0 and 1 for all $i \in [N]$ and $x \in \{0,1\}$.

Theorem \ref{th: minimax} gives conditions under which $\hat{L}$ attains the minimax expected squared error for estimation of $L$. It is proven in Appendix \ref{sec: proof minimax}.

\begin{theorem} \label{th: minimax}
Assume a sequence of experiments in which $Y$ is given by \eqref{eq: f} and Assumption \ref{as: random sampling} holds, with $\rho$ constant as $N \rightarrow \infty$. Let $L$ and $\hat{L}$ be defined respectively by \eqref{eq: L} and \eqref{eq: hat L}, for any vector $u \equiv u(X) \in \mathbb{R}^N$ that is a function of the treatment vector $X$, and let $\mathcal{F}$ denote the space of possible potential outcome mappings $\{f_i\}_{i=1}^N$. Then the $\hat{L}$ attains the best possible worst-case expected squared error,
\begin{align}\label{eq: minimax 1}
\max_{\{f_i\}_{i=1}^N \in \mathcal{F}} \mathbb{E}[(L - \hat{L})^2] = \min_{\hat{L} \in \mathcal{L}} \ \ \max_{\{f_i\}_{i=1}^N \in \mathcal{F}} \ \ \mathbb{E}[(L - \hat{L})^2]
\end{align}
Here $\mathcal{L}$ denotes the space of possible estimators taking $(X,Y)$ as input and returning an estimate for $L$.
\end{theorem}






\section{Example}

\cite{miguel2004worms} considers a experiment that took place in a region of Kenya with high  infection rates for intestinal helminths (such as hookworm). 50 participating schools were quasi-randomized into 2 groups. Schools in group 1 received free deworming treatment in year 1998, while the others received treatment in later years. All students were surveyed in year 1999 on various outcomes including parasitic infection ($N = 1969$).\footnote{In fact, 75 schools were divided into 3 groups by ordering the schools by zone and school name, and then assigning every 3rd school to the same group. Group 3 and also school 133 in group 2 were excluded from the 1999 outcome survey, and were removed from our analysis.}

To demonstrate our methods, we will assume that treatment was assigned by sampling schools without replacement after stratifying by administrative zone. For each student $i \in [N]$, let $X_i \in \{0,1\}$ and $W_i \in \{0,1\}$ respectively denote their direct and indirect treatment status, where $X_i=1$ if $i$'s school was assigned to group 1, and $W_i=1$ if at least half of the participating schools in a 6 km radius of $i$'s school were assigned to group 1, not counting $i$'s own school. Let $Y_i=0$ denote the outcome that student $i$ reported no parasitic infections in the 1999 survey, with $Y_i=1$ indicating otherwise. Let $\theta_i^*$ denote unit $i$'s outcome under a counterfactual of school-level isolated treatment -- that is, under the counterfactual in which $i$'s school received its assigned treatment, while all other schools received the control.

Figure \ref{fig: worms data} and Table \ref{table: worms data} show counts of schools and average outcomes as a function of direct and indirect treatment. Grouping the units by the indirect treatment $W$ produced large differences in outcome: among students with high levels of indirect treatment ($W_i=1$), only 19\% of students reported infections in the outcome survey, compared to 42\% of students with low levels of indirect treatment. This difference in outcomes is slightly larger than the difference of 22 percentage points that arises when grouping by $X$, the direct treatment variable. When grouping by $W$, both groups had the same proportion of schools that received direct treatment, as shown in Table \ref{table: worms data}.

The conservative point estimate $\hat{L}$ equaled 190 individuals (9.6\%), suggesting that a substantial fraction of students were affected by the treatment of a school other than their own. On the other hand, interval estimates based on Theorem \ref{th: coverage} suggest high levels of uncertainty, although statistical significance was retained: the 90\% level one-sided interval for $\psi$ was $[1.3\%,\, 100\%]$, and the 95\% level one-sided interval was $[0.3\%,\, 100\%]$. 

These estimates are weaker in interpretation and have greater uncertainty than those of \cite{miguel2004worms}. However, their analysis assumed knowledge of which schools are affected by each other, and ruled out the possibility of longer range interactions. Such interactions might nonetheless occur; for example, an infection in one school might spread to neighboring schools, and then from the neighboring schools to more distant ones. In a sense, our results complement theirs, suggesting that while the particulars of the earlier spillover analysis might be debatable due to model uncertainty, the existence and prevalence of such spillovers may be credible under much weaker assumptions.

\begin{table}
\begin{center}
\begin{tabular}{ l | r r }
 Treatment & \# Schools & Average Outcome \\
 \hline
 $W_i=1, X_i=1$ & 10 & 0.14 \\ 
 $W_i=1, X_i=0$ & 10 & 0.22 \\  
 $W_i=0, X_i=1$ & 15 & 0.23 \\
 $W_i=0, X_i=0$ & 14 & 0.53
\end{tabular}
\end{center}
\caption{School counts and average outcomes}
\label{table: worms data}
\end{table}

\begin{figure}
\begin{center}
\includegraphics[width=2.7in]{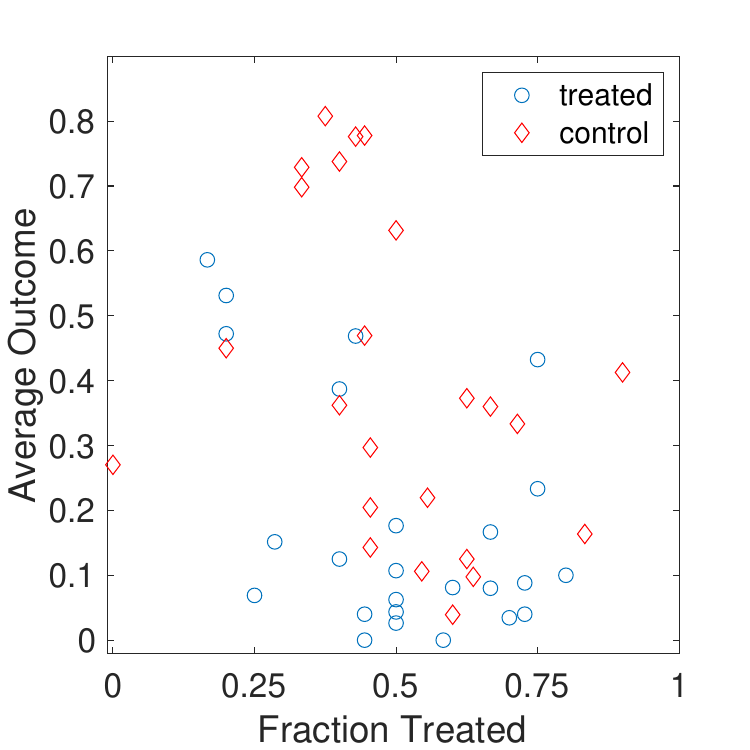}
\end{center}
\caption{Fraction of students with at least 1 infection in each school, versus fraction of treated neighboring schools. Direct treatment of schools shown by marker shape.}
\label{fig: worms data}
\end{figure}

\appendix

\section{Proof of Theorem \ref{th: clt}} \label{sec: CLT proof}

\subsection{Preliminaries}

Without loss of generality, we may assume the sample of treated units to be drawn sequentially -- i.e., starting with an urn containing $N$ balls, each labeled to correspond to a unit in the experiment, we may draw balls one by one without replacement until $T$ are drawn. Let $J=(J_1,\ldots,J_T)$ denote the treated units in the order that they are drawn, so that the treatment vector $X$ depends on $J$ by
\begin{equation}\label{eq: J to X}
 X_i = 1\{i \in J\}, \qquad \forall \ i \in [N].
\end{equation}
Then given a statistic $\tau = \sum_{i=1}^N v_i(X)$ as required by Theorem \ref{th: clt}, 
we may write $v_i(J)$ to denote $v_i(X)$ with $X$ given by \eqref{eq: J to X}.

Given the adjacency matrix $G$, let $\eta_i$ denote the neighbors of unit $i$, and let $\Gamma_i$ denote the units to whom $i$ is a neighbor,
\begin{align} \label{eq: eta gamma}
\eta_i &= \{j: G_{ij} = 1\}, & \Gamma_i & = \{j: G_{ji} = 1\}, & i \in [N]
\end{align}
and let $\bar{\eta}_i$ and $\bar{\Gamma}_i$ denote the union of unit $i$ with $\eta_i$ or $\Gamma_i$ 
\begin{align} \label{eq: bar eta gamma}
\bar{\eta}_i & = i \cup \eta_i, & \bar{\Gamma}_i & = i \cup \Gamma_i, & i \in [N]
\end{align}
so that $\bar{\eta}_i$ is the set of units whose treatments affect the value of $v_i(X)$, and $\bar{\Gamma}_i$ is the set of functions in $\{v_j\}_{j=1}^N$ that are affected by the treatment of unit $i$.

Given positive integers $a < b$, we will use the notation $j_{a:b}$ to denote a vector $(j_a, j_{a+1},\ldots, j_b)$. For example, we can write $J_{1:t}$ to denote $(J_1,\ldots,J_t)$, and $J_{t+1:T} = (J_{t+1},\ldots, J_T)$. Vectors may be concatenated; for example, given $J_{1:t-1}$ and a scalar $a$, we can write $(J_{1:t-1},\, a)$ to denote their $t$-dimensional concatenation.

Set operations involving a vector will implicitly convert the vector to the set of its unique values. For the vector $J$ that is generated by sampling without replacement, this results in the set $\{J_1,\ldots,J_T\}$. For example, we can write $\bar{\eta}_i \cap J$, $\bar{\eta}_i / J$, and $\bar{\eta}_i\, \Delta\, J$ to respectively denote the intersection, set difference, and symmetric difference of the sets $\bar{\eta}_i$ and $\{J_1,\ldots,J_T\}$. Similarly, we can write $J_t \notin J_{1:t-1}$ to denote that $J_t$ may not be a member of the set $\{J_1,\ldots,J_{t-1}\}$. 

Using these conventions, the bounded dependence condition \eqref{eq: v local} may be written for each function $v_i$ as
\begin{equation} \label{eq: local 2}
 v_i(J) = v_i(\bar{\eta}_i \cap J),
\end{equation}
meaning that the value of $v_i$ depends only on the treatment status of the units in $\bar{\eta}_i$.

We will let $\emptyset$ denote the empty set.

\subsection{Proof of Theorem \ref{th: clt}}
Our proof will use a martingale central limit theorem found in 
\cite[Thm 35.12]{billingsley2013convergence}: 

\begin{theorem} \label{th: martingale clt}
Given a Martingale array $M_{nt}$, with differences $Z_{nt} = M_{nt} - M_{n,t-1}$ and conditional variances 
$\sigma_{nt}^2 = E[Z_{nt}^2| \mathcal{F}_{n,t-1}]$, let the following hold:
\begin{enumerate}[label=\text{(C\arabic*)}]
    \item For all $n$, let $M_{nt}$ and $\sum_{r=1}^t \sigma_{nr}^2$ converge almost surely to finite limits as $t \rightarrow \infty$. \label{item: bounded limits}
    \item For all $\epsilon > 0$, it holds that $\lim_{n \rightarrow \infty} \sum_{t=1}^\infty E[Z_{nt}^2 \cdot 1\{|Z_{nt}| > \epsilon\}] = 0$ \label{item: bounded increments}
    \item $\sum_{t=1}^\infty \sigma_{nt}^2 \rightarrow_p 1$ as $n \rightarrow \infty$. \label{item: convergence of variance}
\end{enumerate}
Then $\sum_{t=1}^\infty Z_{nt}$ converges in distribution to $N(0,1)$. 
\end{theorem}

We will apply Theorem \ref{th: martingale clt} by the following construction. Given an experiment involving $N$ units, with $J = (J_1,\ldots,J_T)$, let $M_{N0}=\mathbb{E}\tau/\stdt$ and let $M_{N1},\ldots,M_{NT}$ denote the conditional expectation of $\tau$ as entries of $J$ are sequentially revealed,
\[ M_{Nt} = \frac{E[\, \tau\, |\,  J_{1:t}\, ]}{\stdt}.\]
It can been seen that $\{M_{Nt}\}_{t=0}^T$ is a martingale and that $M_{NT} = \tau/\stdt$. For $t \in [T]$, let $Z_{Nt} = M_{Nt} - M_{N,t-1}$ denote the martingale difference, and let $\sigma_{Nt}^2 = \mathbb{E}[Z_{Nt}^2 | J_{1:t-1}]$ denote the variance of $Z_{Nt}$ given $J_{1:t-1}$. To extend the martingale beyond time $T$, for all $t>T$ let $M_{Nt} = M_{NT}$, $Z_{Nt} = 0$, and $\sigma_{Nt}^2 = 0$.




Lemma \ref{le: bounded increments} states that the differences $Z_{Nt}$ are bounded and go to zero as $\operatorname{Var} \tau$ increases. It will be used to establish Condition \ref{item: bounded increments} of Theorem \ref{th: martingale clt}.

\begin{lemma}\label{le: bounded increments}
Under the conditions of Theorem \ref{th: clt}, there exists $C_1$ such that for $t \in [T]$, it holds that $|Z_{Nt}| \leq C_1/\stdt$, where the value of $C_1$ is fixed as $N\rightarrow \infty$.
\end{lemma}

\noindent To establish condition \ref{item: convergence of variance} of Theorem \ref{th: martingale clt}, we will require a martingale version of the method of bounded differences (Azuma-Hoeffding), taken from \cite[Cor. 5.1]{dubhashi2009concentration}.

\begin{theorem}\label{th: martingale bounded diff}
Given a random vector $J = (J_1,\ldots,J_T)$ whose entries need not be independent, let the function $f$ satisfy the bounded differences condition
\begin{align} \label{eq: bounded diff}
| E[f(J)|J_{1:\ell} = j_{1:\ell}] - E[f(J)|J_{1:\ell-1}=j_{1:\ell-1}, J_\ell= j_\ell']| \leq c_\ell,
\end{align}
for all $\ell \in [T]$, $j_{1:\ell}$, and $j_\ell'$. Then $\mathbb{P}(|f - Ef| \geq \epsilon) \leq 2e^{-2\epsilon^2/(\sum_{\ell=1}^T c_\ell^2)}$.
\end{theorem}

\noindent Lemma \ref{le: bounded diff} states a bounded difference property for $\sum_t \sigma_{Nt}^2$, allowing Theorem \ref{th: martingale bounded diff} to be used to establish condition \ref{item: convergence of variance}.

\begin{lemma}\label{le: bounded diff}
    Let $f(J) = \sum_{t=1}^T \sigma_{Nt}^2$. Under the conditions of Theorem \ref{th: clt}, there exists $C_2$ such that for all $\ell$, $j_{1:\ell}$, and $j_{\ell}' \notin j_{1:\ell-1}$ the function $f$ satisfies the bounded differences property
    \[|E[f(J)|J_{1:\ell} = j_{1:\ell}] - E[f(J)|J_{1:\ell} = (j_{1:\ell-1}, j_\ell')]| \leq \frac{C_2}{\operatorname{Var} \tau},\]
    where the value of $C_2$ is fixed as $N \rightarrow \infty$.
\end{lemma}

\noindent Lemmas \ref{le: bounded increments} and \ref{le: bounded diff} are proven in Section \ref{sec: lemmas 1 and 2}.

\begin{proof}[Proof of Theorem \ref{th: clt}]
    It can be seen that $\sum_{t=1}^\infty Z_{Nt} = \sum_{t=1}^T Z_{Nt} = (\tau - \mathbb{E}\tau)/\stdt$. We show that the conditions \ref{item: bounded limits}, \ref{item: bounded increments}, and \ref{item: convergence of variance} listed in the statement of Theorem \ref{th: martingale clt} hold:
    \begin{enumerate}
    \item Condition \ref{item: bounded limits}, which requires $M_{Nt}$ and $\sum_{r=1}^t \sigma_{Nr}^2$ to converge to finite limits, holds trivially since $M_{Nt} = M_{NT}$ for all $t > T$.  
    \item Lemma \ref{le: bounded increments} together with the assumption that $\operatorname{Var} \tau \geq N^{1/2 + \delta}$ implies condition \ref{item: bounded increments}.
    \item By Lemma \ref{le: bounded diff} and the assumption that $\operatorname{Var} \tau \geq N^{1/2 + \delta}$, we may apply Theorem \ref{th: martingale bounded diff} to the function $f(J) = \sum_{t=1}^T \sigma_{Nt}^2$, with $c_t = C_2 N^{-1/2-\delta}$ for $t \in [T]$. It follows that
    \[\sum_{t=1}^T c_t^2 = C_2^2 T N^{-1-2\delta} \leq C_2^2N^{-2\delta},\]
    where the inequality holds because $T < N$. It hence follows from Theorem \ref{th: martingale bounded diff} for any $\epsilon>0$  that
    \begin{equation}\label{eq: clt 0}
  \mathbb{P}(|f - \mathbb{E}f| \geq \epsilon) \leq 2e^{-2\epsilon^2\Omega(N^{2\delta})} \rightarrow 0,
  \end{equation}
implying convergence in probability of $\sum_{t=1}^\infty \sigma_{Nt}^2$ to its expectation $\mathbb{E}[f]$, which can be shown to equal 1 by the following steps:
\begin{align}
\nonumber \mathbb{E}[f] & = \mathbb{E}\left[ \sum_{t=1}^T \sigma_{Nt}^2\right] \\
&= \mathbb{E}\left[ \sum_{t=1}^T Z_{Nt}^2\right] \label{eq: clt 1}\\
& = \mathbb{E}\left[ \left(\sum_{t=1}^T Z_{Nt}\right)^2 \right] \label{eq: clt 2}\\ 
& = \mathbb{E}[ (\tau - \mathbb{E}\tau)^2] / \operatorname{Var} \tau = 1 \label{eq: clt 3}
\end{align}
where \eqref{eq: clt 1} holds by tower property of expectation, as 
\begin{align*}
\mathbb{E}[ \sigma_{Nt}^2] &= \mathbb{E}\left[ \mathbb{E}[Z_{Nt}^2|J_{1:t-1}]\right] = \mathbb{E}\left[ Z_{Nt}^2\right],
\end{align*}
\eqref{eq: clt 2} holds as the martingale differences are uncorrelated; and \eqref{eq: clt 3} holds as $\sum_{t=1}^T Z_{Nt} = (\tau - \mathbb{E}\tau)/\stdt$. As \eqref{eq: clt 3} and \eqref{eq: clt 0} together imply that $\sum_{t=1}^T \sigma_{Nt}^2$ converges in probability to 1, this establishes Condition \ref{item: convergence of variance}.
\end{enumerate}
As \ref{item: bounded limits}, \ref{item: bounded increments}, and \ref{item: convergence of variance} have been shown to hold, Theorem \ref{th: martingale clt} implies that $M_{NT} - M_{N0} = (\tau - \mathbb{E}\tau)/\stdt$ is normally distributed.
\end{proof}

\subsection{Proof of Lemmas \ref{le: bounded increments} and \ref{le: bounded diff}} \label{sec: lemmas 1 and 2}

To prove Lemmas \ref{le: bounded increments} and \ref{le: bounded diff}, we will use the following identity, which states that changing the set of treated units does not affect the conditional expectation of $v_i$ is the change is retricted to units outside of $\bar{\eta}_i$.

\begin{lemma}\label{le: basic prob}
 If $(j_{1:t} \Delta j_{1:t}') \cap \bar{\eta}_i = \emptyset$ then
 \begin{equation}\label{eq: basic prob}
 \mathbb{E}[v_i(J)|J_{1:t} = j_{1:t}] = \mathbb{E}[v_i(J)| J_{1:t} = j_{1:t}'] 
 \end{equation}
\end{lemma}

\begin{proof}[Proof of Lemma \ref{le: basic prob}]
Given $j_{1:t}$ and $j_{1:t}'$, let their intersection with $\bar{\eta}_i$ be denoted respectively by $S_0$ and $S_0'$,
\begin{align*}
S_0 & = \bar{\eta}_i \cap j_{1:t} & S_0' & = \bar{\eta}_i \cap j_{1:t}' 
\end{align*}
Since $v_i$ depends only on $J \cap \bar{\eta}_i$, the claim \eqref{eq: basic prob} follows from
\begin{align}
\nonumber \mathbb{E}[v_i(J)| J_{1:t} = j_{1:t}] & = \sum_{S \subseteq \bar{\eta}_i} v_i(S)\cdot P(J \cap \bar{\eta}_i = S | J_{1:t} = j_{1:t}) \\
& = \sum_{S_1 \subseteq \bar{\eta}_i/ j_{1:t}} v_i(S_0 \cup S_1) P(J_{t+1:T} \cap \bar{\eta}_i = S_1| J_{1:t} = j_{1:t}) \label{eq: basic prob decomp} \\
 & = \sum_{S_1 \subseteq \bar{\eta}_i/j_{1:t}} v_i(S_0 \cup S_1) \frac{{{N - t - |\bar{\eta}_i/j_{1:t}|} \choose {T - t - |S / j_{1:t}|}}}{{{N - t} \choose {T - t}}}  \label{eq: basic prob counting} \\
  & = \sum_{S_1 \subseteq \bar{\eta}_i / j_{1:t}'} v_i(S_0' \cup S_1)  \frac{{{N - t - |\bar{\eta}_i/j_{1:t}'|} \choose {T - t - |S / j_{1:t}'|}}}{{{N - t} \choose {T - t}}}, \label{eq: basic prob substitute} \\
  \nonumber & = \sum_{S \subseteq \bar{\eta}_i} v_i(S)\cdot P(J \cap \bar{\eta}_i = S | J_{1:t} = j_{1:t}') \\
  \nonumber & = \mathbb{E}[v_i(J)|J_{1:t} = j_{1:t}'] 
\end{align}
where \eqref{eq: basic prob decomp} decomposes $S$ into $S_0 \cup S_1$, implying $S_1 = S / S_0$; \eqref{eq: basic prob counting} holds by counting arguments\footnote{The numerator ${{N - t - |\bar{\eta}_i/j_{1:t}|} \choose {T - t - |S / j_{1:t}|}}$ counts the ways to assign the remaining $T - t$ draws such that $S/j_{1:t}$ is included while $\bar{\eta}_i / S$ is excluded.}, and \eqref{eq: basic prob substitute} holds because $S_0 = S_0'$ and because $S/ j_{1:t} = S/ j_{1:t}'$ for all $S \subseteq \bar{\eta}_i$, both of which follow from the assumption that $(j_{1:t}\, \Delta j_{1:t}') \cap \bar{\eta}_i = \emptyset$. 
\end{proof}





\subsubsection{Proof of Lemma \ref{le: bounded increments}} The proof of Lemma \ref{le: bounded increments} is given below:

\begin{proof}[Proof of Lemma \ref{le: bounded increments}] 
We first show that if $J_t \notin \bar{\eta}_i$, then 
\begin{equation}\label{eq: bounded increments bound}
\mathbb{E}[v_i|J_{1:t} ] - \mathbb{E}[v_i|J_{1:t-1} ] \leq \frac{2C|\bar{\eta}_i|}{N - t}
\end{equation}
To show \eqref{eq: bounded increments bound}, we observe that
    \begin{align}
    \nonumber & \big|\ \mathbb{E}[v_i|J_{1:t}=j_{1:t} ] - \mathbb{E}[v_i|J_{1:t-1} = j_{1:t-1}]\ \big| \\
     & = \left|\ \mathbb{E}[v_i|J_{1:t}=j_{1:t}] - \sum_{j_t' \notin j_{1:t-1}} \frac{\mathbb{E}[v_i|J_{1:t} = (j_{1:t-1}, j_t')]}{N - t + 1}  \ \right| \label{eq: bounded increments expectation}\\
    \nonumber & \leq \sum_{j_t' \notin j_{1:t-1}} \left| \ \frac{ \mathbb{E}[v_i|J_{1:t}=j_{1:t}] - \mathbb{E}[v_i|J_{1:t} = (j_{1:t-1}, j_t')] } {N - t + 1} \ \right|  \\
      \nonumber & = \sum_{j_t \notin j_{1:t-1}} \Big(1\{j_t' \notin \bar{\eta}_i\} + 1\{j_t' \in \bar{\eta}_i\}\Big) \left|\ \frac{ \mathbb{E}[v_i|J_{1:t}=j_{1:t}] - \mathbb{E}[v_i|J_{1:t} = (j_{1:t-1}, j_t')] } {N - t + 1} \ \right|\\
    & \leq \frac{2C|\bar{\eta}_i|}{N-t}  \label{eq: bounded increments final}
    \end{align}
    where \eqref{eq: bounded increments expectation} is the law of total expectation, as $j_t'$ is uniformly distributed on the set $[N]/j_{1:t-1}$;  and \eqref{eq: bounded increments final} involves the contribution of two factors:
    \begin{enumerate} 
        \item For all $j_t' \notin \bar{\eta}_i$ it holds that $(j_t, j_t') \notin \bar{\eta}_i$ and hence that
        \[\big(\, j_{1:t}\, \Delta\, (j_{1:t-1},j_t) \, \big) \cap \bar{\eta}_i = \emptyset\] 
        therefore by Lemma \ref{le: basic prob} it follows that
\[\sum_{j_t'\notin j_{1:t-1}} 1\{j_t' \notin \bar{\eta}_i\} \ \Big|\, \mathbb{E}[v_i(J)|J_{1:t}=j_{1:t}] - \mathbb{E}[v_i(J)| J_{1:t-1} = (j_{1:t-1}, j_{t}')]\, \Big| = 0,\]
\item For all $j_t' \in \bar{\eta}_i$ we may apply $|v_i| \leq C$ (by Assumption \ref{as: tau clt}), implying that 
    \[\sum_{j_t'\notin j_{1:t-1}} 1\{j_t' \in \bar{\eta}_i\} \ \left|\, \frac{\mathbb{E}[v_i(J)|J_{1:t} = j_{1:t}] - \mathbb{E}[v_i(J)| J_{1:t} = (j_{1:t-1}, j_{t}')]}{N - t +1}\,\right| \leq \frac{2C |\bar{\eta}_i|}{N - t}.\]
\end{enumerate}
Having shown \eqref{eq: bounded increments bound}, we observe that
\begin{align}
\nonumber |Z_t| & = \left| \frac{\mathbb{E}[ \tau| J_{1:t}] - \mathbb{E}[\tau|J_{1:t-1}] }{\stdt} \right| \\
& \leq  \sum_{i=1}^N \left| \frac{ \mathbb{E}[ v_i|J_{1:t}] - \mathbb{E}[v_i|J_{1:t-1}]}{\stdt} \right| \label{eq: bounded increments tau} \\
\nonumber & = \sum_{i=1}^N \left(1\{J_t \notin \bar{\eta}_i\} + 1\{J_t \in \bar{\eta}_i\}\right) \left|\frac{ \mathbb{E}[ v_i|J_{1:t}] - \mathbb{E}[v_i|J_{1:t-1}]  } {\stdt} \right|   \\
& \leq \frac{1}{\stdt} \left(\frac{2NCd_{\max}}{N-t} + 2Cd_{\max}\right) \label{eq: bounded increments main} \\
& \leq \frac{C_1}{\stdt} \label{eq: bounded increments var}
\end{align}
where \eqref{eq: bounded increments tau} holds because $\tau = \sum_{i=1}^N v_i$; \eqref{eq: bounded increments main} holds because \eqref{eq: bounded increments bound} bounds the difference $\mathbb{E}[ v_i|J_{1:t}] - \mathbb{E}[v_i|J_{1:t-1}]$ for all units $i$ such that $J_t \notin \bar{\eta}_i$, while Assumption \ref{as: tau clt} enforces $|v_i|\leq C$, which bounds the difference for the remaining units which satisfy $J_t \in \bar{\eta}_i$ and hence belong to the set $\bar{\Gamma}_{J_t}$ (whose cardinality is bounded by Assumption \ref{as: degree}); and \eqref{eq: bounded increments var} holds for $C_1 = 2Cd_{\max}/(1-\rho) + 2Cd_{\max}$ because $N/(N-t) \geq 1/(1-\rho)$ by Assumption \ref{as: random sampling}.

\end{proof}

\subsubsection{Proof of Lemma \ref{le: bounded diff}} 
Before proving this lemma, we introduce $J^{(\ell)}$, a coupled version of $J$ that is generated as follows. Given $J_{1:\ell}$, let $j_\ell'$ be generated by sampling without replacement from $[N]/J_{1:\ell-1}$, and let $h:[N]\mapsto[N]$ denote the mapping that swaps $J_\ell$ for $j_\ell'$ and vice versa, otherwise leaving the argument unchanged:
    \[ h(j) = \begin{cases} j_\ell' & \textup{if } j = J_\ell \\
    J_\ell & \textup{if } j = j_\ell' \\
    j & \textup{otherwise} \end{cases} \]
and let $J^{(\ell)} = h(J)$ denote the vector in $\mathbb{R}^T$ that results when $h$ is applied elementwise to $J$. It can be seen that $J$ and $J^{(\ell)}$ satisfy the following:
\begin{enumerate}
	\item Conditioned on $J_{1:\ell}$ and $j_\ell'$, $J_{\ell+1:T}^{(\ell)}$ has the same distribution as $J_{\ell+1:T}$ would have if $J_\ell$ was changed to $j_\ell'$:
    \[P(J_{\ell+1:T}^{(\ell)} = j_{\ell+1:T}|J_{1:\ell} = j_{1:\ell}, j_\ell') = P(J_{\ell+1:T} = j_{\ell+1:T}| J_{1:\ell} = (j_{1:\ell-1}, j_\ell'))\]
    To prove this, let $C_P = \frac{(N-\ell)!}{(N-T)!}$ denote the number of ways to draw $T-\ell$ items from a population of size $N-\ell$, and observe that
    \begin{align}
     \nonumber P(J^{(\ell)}_{\ell+1:T} = j_{\ell+1:T}|J_{1:\ell}=j_{1:\ell}, j_\ell') & = P(J_{\ell+1:T} = h^{-1}(j_{\ell+1:T})|J_{1:\ell} = j_{1:\ell}, j_\ell')  \\ 
    \nonumber & = \frac{1}{C_P}  1\{h^{-1}(j_{\ell+1:T}) \cap j_{1:\ell} = \emptyset\} \\ 
     \nonumber & = \frac{1}{C_P} 1\{j_{\ell+1:T} \cap h(j_{1:\ell}) = \emptyset\} \\
     \nonumber & = P(J_{\ell+1:T} = j_{\ell+1:T}|J_{1:\ell} = h(j_{1:\ell})) \\ 
\nonumber & = P(J_{\ell+1:T} = j_{\ell+1:T}| J_{1:\ell} = (j_{1:\ell-1}, j_\ell'))
    \end{align}
    where the first line holds because $h$ is invertible; and the second and fourth lines hold because because $J_{\ell+1:T}$ conditioned on $J_{1:\ell}$ is equally probable to take any value in its support (which has cardinality $C_P$ for all $J_{1:\ell}$).
    \item For all $t$, the symmetric difference $J_{1:t} \Delta J_{1:t}^{(\ell)}$ equals either $\{J_\ell, J_\ell^{(\ell)}\}$ or $\emptyset$
    \item Conditioning on $J_{1:t}^{(\ell)}$ provides no additional information regarding $J_{t+1:T}$ given $J_{1:t}$ for $t \geq \ell$. In other words, it holds for $ t \geq \ell$ and any function $f$ of $J$ that
\begin{align*}
\E[f(J) | J_{1:t} ] & = \E[f(J) | J_{1:t}, J_{1:t}^{(\ell)}] \\
\E[f(J^{(\ell)}) | J_{1:t}^{(\ell)} ] & = \E[f(J^{(\ell)}) | J_{1:t}^{(\ell)}, J_{1:t}]. 
\end{align*}
\end{enumerate}
\noindent Given $J^{(\ell)}$, let $v_i^{(\ell)}$, $\tau^{(\ell)}$, $M_{Nt}^{(\ell)}$, and $Z_{Nt}^{(\ell)}$ be defined as follows
\begin{align*}
v_i^{(\ell)} & = v_i(J^{(\ell)}) & \tau^{(\ell)} &= \sum_{i=1}^N v_i^{(\ell)} \\
M_{Nt}^{(\ell)} & = \E[\tau^{(\ell)}| J_{1:t}^{(\ell)}] & Z_{Nt}^{(\ell)} & =  M_{Nt}^{(\ell)} - M_{N,t-1}^{(\ell)},
\end{align*}
and let $\diff_{it}$ be given by
\begin{equation*}
\diff_{it} = \E[v_i|J_{1:t}] - \E[v_i|J_{1:t-1}] - \E[v_i^{(\ell)}|J_{1:t}^{(\ell)}] + \E[v_i^{(\ell)}|J_{1:t-1}^{(\ell)}]
\end{equation*}
Lemma \ref{le: diff}, which is proven in Section \ref{sec: le diff}, gives a bound on the conditional expectation of $\diff_{it}$. 

\begin{lemma}\label{le: diff}
Let the conditions of Theorem \ref{th: clt} hold. There exists $C_3$ such that for all $i\in [N]$, $t \geq \ell+1$, and $(j_{1:t-1}, j_{1:t-1}')$ in the joint support of $(J_{1:t-1}, J_{1:t-1}^{(\ell)})$ conditioned on $(J_{1:\ell}, J_{1:\ell}')$, it holds that if $\bar{\eta}_i \cap \{j_\ell, j_\ell'\} = \emptyset$, then
\[\E\Big[ |\diff_{it}|\, \Big| J_{1:t-1} = j_{1:t-1}, J_{1:t-1}^{(\ell)} = j_{1:t-1}' \, \Big] = 0\]
and otherwise
\[ \E\Big[ |\diff_{it}|\, \Big| J_{1:t-1} = j_{1:t-1}, J_{1:t-1}^{(\ell)} = j_{1:t-1}' \, \Big]  
\leq \frac{C_3 d^{\max}}{N-t}
\]
\end{lemma}

\begin{proof}[Proof of Lemma \ref{le: bounded diff}]
Given $j_{1:\ell}$ and $j_\ell' \notin j_{1:\ell-1}$, let $j_{1:\ell}' = (j_{1:\ell-1}, j_\ell')$ denote the vector that replaces the last entry of $j_{1:\ell}$ with $j_\ell'$. Let $\mathcal{I}$ denote the set of units whose neighborhoods include $j_\ell$ or $j_\ell'$,
\[ \mathcal{I} = \{i \in [N]: \bar{\eta}_i \cap (j_\ell, j_\ell') \neq \emptyset\}\]
The lemma follows from the following chain of equations
\begin{align}
\nonumber & \Bigg| \E\left[ \sum_{t=1}^T \E[Z_{Nt}^2| J_{1:t-1}]\ \bigg| J_{1:\ell} = j_{1:\ell}\right] - 
\E\left[ \sum_{t=1}^T \E[Z_{Nt}^2| J_{1:t-1}]\ \bigg| J_{1:\ell} = j_{1:\ell}' \right] \Bigg|\\
& = \Bigg| \E\left[ \sum_{t=\ell+1}^T \E[Z_{Nt}^2| J_{1:t-1}]\ \bigg| J_{1:\ell} = j_{1:\ell}\right] - 
\E\left[ \sum_{t=\ell+1}^T \E[Z_{Nt}^2| J_{1:t-1}]\ \bigg| J_{1:\ell} = j_{1:\ell}' \right] \Bigg| \label{eq: t before ell}\\
& = \Bigg| \E\left[ \sum_{t=\ell+1}^T \E[Z_{Nt}^2| J_{1:t-1}]\ \bigg| J_{1:\ell} = j_{1:\ell}\right] - 
\E\left[ \sum_{t=\ell+1}^T \E[(Z_{Nt}^{(\ell)})^2| J_{1:t-1}^{(\ell)}]\ \bigg| J_{1:\ell}^{(\ell)} = j_{1:\ell}' \right] \Bigg| \label{eq: zprime} \\
 & = \Bigg| \sum_{t=\ell+1}^T \E\left[  \E\Big[Z_{Nt}^2 - (Z_{Nt}^{(\ell)})^2\ \Big| J_{1:t-1}, J_{1:t-1}^{(\ell)}\Big]\ \bigg| J_{1:\ell} = j_{1:\ell}, J_{1:\ell}^{(\ell)}=j_{1:\ell}'\right] \Bigg| \label{eq: J and Jprime}\\
& = \Bigg| \sum_{t=\ell+1}^T \E\left[  \E\Big[(Z_{Nt} + Z_{Nt}^{(\ell)})(Z_{Nt} - Z_{Nt}^{(\ell)})\ \Big| J_{1:t-1}, J_{1:t-1}^{(\ell)}\Big]\ \bigg| J_{1:\ell} = j_{1:\ell}, J_{1:\ell}^{(\ell)}=j_{1:\ell}' \right] \Bigg| \label{eq: a^2 - b^2} \\
 & \leq  \frac{2C_1}{\stdt} \sum_{t=\ell+1}^T \E\left[  \E\Big[\ |Z_{Nt} - Z_{Nt}^{(\ell)}|\ \Big| J_{1:t-1}, J_{1:t-1}^{(\ell)}\Big]\ \bigg| J_{1:\ell} = j_{1:\ell}, J_{1:\ell}^{(\ell)}=j_{1:\ell}' \right] \label{eq: factor out a+b} \\
& = \frac{2C_1}{\stdt}\sum_{t=\ell+1}^T \E\Bigg\{  \E\Bigg[\ 
\bigg| \E[\tau|J_{1:t}] - \E[\tau|J_{1:t-1}] - \E[\tau^{(\ell)}|J_{1:t}^{(\ell)}] + \E[\tau^{(\ell)}|J_{1:t-1}^{(\ell)}] \bigg| \label{eq: substitute tau} \\
\nonumber & \hskip3cm \Bigg| J_{1:t-1}, J_{1:t-1}^{(\ell)}\Bigg]\ \Bigg| J_{1:\ell} = j_{1:\ell}, J_{1:\ell}^{(\ell)}=j_{1:\ell}'\Bigg\} \\
& \leq \frac{2C_1}{\operatorname{Var} \tau}\sum_{t=\ell+1}^T \sum_{i=1}^N \E\Bigg\{  \E\Bigg[\ 
 \bigg|  \Big( \E[v_i|J_{1:t}] - \E[v_i|J_{1:t-1}] - \E[v_i^{(\ell)}|J_{1:t}^{(\ell)}] + \E[v_i^{(\ell)}|J_{1:t-1}^{(\ell)}] \Big) \bigg|  \label{eq: substitute vi}\\
\nonumber & \hskip3cm \Bigg| J_{1:t-1}, J_{1:t-1}^{(\ell)}\Bigg]\ \Bigg| J_{1:\ell} = j_{1:\ell}, J_{1:\ell}^{(\ell)} = j_{1:\ell}'  \Bigg\} \\
 & = \frac{2C_1}{\operatorname{Var} \tau}\sum_{t=\ell+1}^T \sum_{i=1}^N \E\Big\{  \E\Big[\ \big|   \diff_{it}  \big|  \Big| J_{1:t-1}, J_{1:t-1}^{(\ell)}\Big]\ \Big| J_{1:\ell} = j_{1:\ell}, J_{1:\ell}^{(\ell)} = j_{1:\ell}'  \Big\} \label{eq: substitute diff} \\
 & \leq \frac{2C_1}{\operatorname{Var} \tau} \sum_{i \in \mathcal{I}} \sum_{t=\ell+1}^T \frac{C_3 d^{\max}}{N-t} \label{eq: apply bounded diff helper lemma}
 \\
 & \leq \frac{C_2}{\operatorname{Var} \tau}\label{eq: bounded diff final}
\end{align}
where \eqref{eq: t before ell} holds since for $t \leq \ell+1$
\[ \E\Big\{ \E[Z_{Nt}^2|J_{1:t-1}] \ \Big| J_{1:\ell} = j_{1:\ell} \Big\} = \E[ Z_{Nt}^2|J_{1:t-1} = j_{1:t-1}],\]
and hence for $t \leq \ell$ it holds that
\[ \E\Big\{ \E[Z_{Nt}^2|J_{1:t-1}] \ \Big| J_{1:\ell} = j_{1:\ell} \Big\} = \E\Big\{ \E[Z_{Nt}^2|J_{1:t-1}] \ \Big| J_{1:\ell} = j_{1:\ell}'\Big\},\]
as $j_{1:t-1} = j_{1:t-1}'$ if $t \leq \ell$; \eqref{eq: zprime} holds as the distribution of $Z_{Nt}^{(\ell)}$ conditioned on $J_{1:\ell}^{(\ell)} = j_{1:\ell}'$ is equal to that of $Z_{Nt}$ conditioned on $J_{1:\ell} = j_{1:\ell}'$; \eqref{eq: J and Jprime} holds since $J^{(\ell)}$ is a deterministic transformation of $J$ given $j_\ell$, so that conditioning on $(J, J^{(\ell)})$ is equivalent to conditioning on $J$ for any function of $J$; \eqref{eq: a^2 - b^2} holds by the identity $(a^2 - b^2) = (a + b)(a-b)$;  \eqref{eq: factor out a+b} holds as $|Z_{Nt} + Z_{Nt}^{(\ell)}| \leq 2C_1/\stdt$ by Lemma \ref{le: bounded increments}; \eqref{eq: substitute tau} - \eqref{eq: substitute diff} hold by definitions of $Z_{Nt}$ and $Z_{Nt}^{(\ell)}$, $\tau$ and $\tau^{(\ell)}$, and $\diff_{it}$ respectively; \eqref{eq: apply bounded diff helper lemma} holds by Lemma \ref{le: diff}; and \eqref{eq: bounded diff final} holds for $C_2$ chosen appropriately because $|\mathcal{I}| \leq |d^{\max}|$ and $T \leq \rho N$. This proves the lemma. 

\end{proof}

\subsection{Proof of Lemma \ref{le: diff}} \label{sec: le diff}

\begin{proof}[Proof of Lemma \ref{le: diff}]

We first establish the following identities
\begin{enumerate}[label=\text{(I\arabic*)}]
    \item If $\bar{\eta}_i \cap \{J_\ell, J_\ell^{(\ell)}\} = \emptyset$, then $\diff_{it} = 0$. \label{item: j_ell}
    \item If $\bar{\eta}_i \cap \{J_\ell, J_\ell^{(\ell)}\} \neq \emptyset$ and $\bar{\eta}_i \cap \{J_t, J_t^{(\ell)}\} = \emptyset$, then $|\diff_{it}| \leq \frac{4Cd^{\max}}{N-t}$ \label{item: j_t}
    \item If $\bar{\eta}_i \cap \{J_\ell, J_\ell^{(\ell)}\} \neq \emptyset$ and $\bar{\eta}_i \cap \{J_t, J_t^{(\ell)}\} \neq \emptyset$, then $|\diff_{it}| \leq 4C$ \label{item: j_ell_t}
\end{enumerate}
by the following arguments:
\begin{enumerate}
    \item To show \ref{item: j_ell}, we observe that $J_{1:t}\, \Delta J_{1:t}^{(\ell)} \subseteq \{J_\ell, J_\ell^{(\ell)}\}$, and hence if $\bar{\eta}_i \cap \{J_\ell, J_\ell^{(\ell)}\} = \emptyset$ then by Lemma \ref{le: basic prob} it holds for $t \geq \ell$ that
\begin{align*} 
\E[v_i|J_{1:t} ] & = \E[v_i^{(\ell)}|J_{1:t}^{(\ell)}] \\
 \E[v_i|J_{1:t-1}] & = \E[v_i^{(\ell)}|J_{1:t-1}^{(\ell)} ]
\end{align*}
implying $\diff_{it} = 0$. 

    \item To show \ref{item: j_t}, we apply \eqref{eq: bounded increments bound} which implies that
\begin{align*} 
\Big|\,\E[v_i|J_{1:t}] - \E[v_i|J_{1:t-1}] \, \Big| \leq \frac{2Cd^{\max}}{N-t} \\
\Big|\, \E[v_i^{(\ell)}|J_{1:t}^{(\ell)}] - \E[v_i^{(\ell)}|J_{1:t-1}^{(\ell)}] \, \Big| \leq \frac{2Cd^{\max}}{N-t}
\end{align*}
implying that $|\diff_{it}| \leq \frac{4Cd^{\max}}{N-t}  $

\item \ref{item: j_ell_t} holds as $|v_i| \leq C$.
\end{enumerate}

\noindent Having shown \ref{item: j_ell}-\ref{item: j_ell_t}, we proceed as follows. For $t > \ell$, given $(j_{1:t}, j_{1:t-1}')$ in the support of $(J_{1:t},\, J^{(\ell)}_{1:t-1})$, let $j_t^*$ be given by
\[ j_t^* = \begin{cases} j_t & \textup{ if } j_t \neq j_\ell' \\ j_\ell & \textup{ if } j_t = j_\ell'\end{cases}\]
so that if $J_{1:t} = j_{1:t}$ and $J_{1:t-1}^{(\ell)} = j_{1:t-1}'$, then $J_{1:t}^{(\ell)} = (j_{1:t-1}', j_t^*)$. Let $\mathcal{J}_i$ denote the set of values for $j_t$ such that $\bar{\eta}_i$ includes either $j_t$ or $j_t^*$,
\[ \mathcal{J}_i = \{ j_t \notin j_{1:t-1}: \bar{\eta}_i \cap \{j_t, j_t^*\} \neq \emptyset \}.\]
It can be seen that $|\mathcal{J}_i| \leq |\bar{\eta}_i|+1$.  

The claim of the Lemma may be shown by
\begin{align}
\nonumber & \E\Big[\, |\diff_{it}|\, \Big| J_{1:t-1} = j_{1:t-1}, J_{1:t-1}^{(\ell)} = j_{1:t-1}' \Big]  \\
\nonumber & = \E\bigg\{\, \E\Big[\, |\diff_{it}|\, \Big| J_{1:t}, J_{1:t}^{(\ell)} \Big] \ \bigg|\, J_{1:t-1} = j_{1:t-1}, J_{1:t-1}^{(\ell)} = j_{1:t-1}' \bigg\} \\
& = \sum_{j_t \notin j_{1:t-1}} \frac{\E\Big[ |\diff_{it}| \Big| J_{1:t} = j_{1:t}, J_{1:t}^{(\ell)} = (j_{1:t-1}', j_t^*)\Big]}{N-t+2} \label{eq: bounded diff expectation}\\
\nonumber & = \sum_{j_t \notin j_{1:t-1}} \Big(1\{j_t \in \mathcal{J}_i\} + 1\{j_t \notin \mathcal{J}_i\}\Big) \frac{\E\Big[ |\diff_{it}| \Big| J_{1:t} = j_{1:t}, J_{1:t}^{(\ell)} = (j_{1:t-1}', j_t^*)\Big]}{N-t+2} \\
& \begin{cases} = 0 & \textup{ if } \bar{\eta}_i \cap (j_\ell, j_\ell')) = \emptyset \\
\leq (C_3d^{\max})/(N-t) & \textup{ else } \end{cases} \label{eq: le diff}
\end{align}
where \eqref{eq: bounded diff expectation} is the law of total expectation, as $j_t$ is uniformly distributed from the set $[N] / j_{1:t-1}$ and $j_t^*$ is deterministically related to $j_t$; and where \eqref{eq: le diff} holds for $C_3$ chosen appropriately, involving the contribution of three factors:
\begin{enumerate}
    \item If $\bar{\eta}_i \cap (j_\ell, j_\ell') = \emptyset$, then $\diff_{it} = 0$ by \ref{item: j_ell} for any value of $j_t$. Taking expectations and summing over $j_t$ yields
    \[ \sum_{j_t \notin j_{1:t-1}} \frac{\E\big[\ |\diff_{it}|\ \big|J_{1:t} = j_{1:t}, J_{1:t}^{(\ell)}=(j_{1:t-1},j_t^*)\big]}{N-t+2} = 0.\]
    \item For the $O(|\bar{\eta}_i|)$ units that are members of $\mathcal{J}_i$, it holds that $\bar{\eta}_i \cap (j_t, j_t^*) \neq \emptyset$. For these members, $|\diff_{it}| \leq 4C$ by \ref{item: j_ell_t}. Taking expectations and summing over the members of $\mathcal{J}_i$ yields
    \[ \sum_{j_t \in \mathcal{J}_i} \frac{\E\big[\ |\diff_{it}|\ \big|J_{1:t} = j_{1:t}, J_{1:t}^{(\ell)}=(j_{1:t-1},j_t^*)\big]}{N - t + 2} \leq (|\bar{\eta}_i| + 1)\frac{4C}{N-t}\] 
    \item For the $O(N)$ units that are non-members of $\mathcal{J}_i$, it holds that $\bar{\eta}_i \cap (j_t, j_t^*) = \emptyset$. For these non-members, $|\diff_{it}| \leq 4Cd^{\max}/(N-t)$ by \ref{item: j_t}. Taking expectations and summing over the non-members of $\mathcal{J}_i$ yields
    \[ \sum_{j_t \notin \mathcal{J}_i} \frac{\E\big[\ |\diff_{it}|\ \big|J_{1:t} = j_{1:t}, J_{1:t}^{(\ell)}=(j_{1:t-1},j_t^*)\big]}{N - t + 2} \leq \frac{4NCd^{\max}}{(N-t)^2}.\] 
\end{enumerate}
\end{proof}

\section{Proof of Theorem \ref{th: coverage}} \label{sec: coverage proof}

To prove Theorem \ref{th: coverage}, we require Lemma \ref{le: V consistency} which gives conditions under which $V$ is consistent for $\operatorname{Var} \tau$. It is proven in the supplemental materials. 

\begin{lemma}\label{le: V consistency}
Let the conditions of Theorem \ref{th: coverage} hold. Then if $\operatorname{Var} \tau \geq N^{1/2 + \delta}$ for some $\delta > 0$, it holds that
\[ V = (\operatorname{Var} \tau)\cdot (1 + o_P(1)) \]
\end{lemma}

\noindent We now prove Theorem \ref{th: coverage}.

\begin{proof}[Proof of Theorem \ref{th: coverage}]
Divide the sequence of experiments into two subsequences, depending on whether or not $\Var \tau \geq N^{1/2 + 0.01}$. We show that \eqref{eq: lower bound} covers $\psi$ for one subsequence, while \eqref{eq: backup optimization} covers $\psi$ for the other subsequence.

\begin{enumerate}
\item For the subsequence where $\Var \tau \geq N^{1/2 + 0.01}$, the conditions of Theorem \ref{th: clt}
are met, and hence $\tau / \sqrt{\operatorname{Var} \tau}$ is asymptotically normal. Since $V = (\operatorname{Var} \tau)\cdot (1 + o_P(1))$ by Lemma \ref{le: V consistency}, it follows that 
\begin{equation}\label{eq: coverage 1}
\frac{\tau}{\sqrt[\uproot{5}+]{V}} \rightarrow_d N(0,1).
\end{equation}
where $\sqrt[\uproot{5}+]{\phantom{.}}$ is given by \eqref{eq: sqrt pos}.
By \eqref{eq: coverage 1}, it follows that \eqref{eq: lower bound} covers $\psi$ with the desired asymptotic probability for this subsequence.
\item For the subsequence where $\Var \tau < N^{1/2 + 0.01}$, by Chebychev's inequality, it holds that 
\[ |\tau| = O_P\left(N^{1/4 + 0.005}\right) = o_P\left(N^{1/3}\right),\]
and hence with probability approaching 1, it holds that
\[ |u^T \theta^*| \leq N^{1/3},\]
in which case $\theta^*$ is in the feasible region of \eqref{eq: backup optimization}. This implies that \eqref{eq: backup optimization} covers $\psi$ with probability asymptotically approaching 1 for this subsequence.
\end{enumerate}
\noindent It follows that taking the minimum of \eqref{eq: lower bound} and \eqref{eq: backup optimization} covers $\psi$ for both subsequences with at least the desired asymptotic probability. This proves the theorem.
\end{proof}

\section{Proof of Theorem \ref{th: consistency}} \label{sec: consistency proof}

\begin{proof}[Proof of Theorem \ref{th: consistency}]

Let $U = (U_1,\ldots,U_N)$ denote a vector of $N$ i.i.d.\! Uniform[0,1] random variates. Then the treatment vector $X$ may be generated by treating the $T$ units with the smallest values in the vector $U$, so that
\[ X_i = \begin{cases} 1 & \textup{ if } U_i \leq U_{[T]} \\ 0 & \textup{ else, }\end{cases}\qquad i \in [N], \]
where $U_{[T]}$ denotes the $T$th empirical quantile of $\{U_i\}_{i=1}^N.$ It can be seen that changing a single entry of $U$ affects at most two entries of the vector $X$, as at most 1 unit goes from untreated to treated and vice versa.

Let $\varepsilon_0$ be given by
\begin{equation} \label{eq: vareps}
 \varepsilon_0 = \sum_{i=1}^N u_i(\theta_i^* - 0.5).
\end{equation}
By Assumption \ref{as: tau}, it holds that
\begin{enumerate}
\item $\mathbb{E}\varepsilon_0 = 0$, as $\E[u_i|X_i] = 0$ for all $i \in [N]$
\item Each entry of $X$ affects at most $d_{\max}+1$ terms in the summation appearing in \eqref{eq: vareps}. To see this, note that changing $X_i$ may affect $\theta_i^*$ and $u_j$ for all units $j$ for whom $i$ is a neighbor; this number is bounded by the out-degree of $G$.
\end{enumerate}

As a result, changing a single entry of $U$ affects at most $2$ entries in $X$, and hence at most $2(d_{\max}+1)$ members of $\{u_i(\theta_i^* - 0.5)\}_{i=1}^N$, which by Assumption \ref{as: tau} have magnitude bounded by $C/2$. It follows that changing a single entry of $U$ changes $\varepsilon_0$ by at most $C(d_{\max}+1)$. Thus Azuma-Hoeffding implies that
\[ P(|\varepsilon_0| \geq \epsilon) \leq \exp\left(\frac{-2\epsilon^2}{NC^2(d_{\max}+1)^2}\right).\]
implying that $|\varepsilon_0| = O_P(N^{1/2})$.

To prove the theorem, we observe that the error $\hat{L} - L$ satisfies
\begin{align}
 \nonumber |\hat{L} - L| & = \left| \frac{|u^T(Y - 0.5)|}{\|u\|_\infty} - \frac{|u^T(Y - \theta^*)|}{\|u\|_\infty} \right|\\
&  \leq \left| \frac{u^T(\theta^* - 0.5)}{\|u\|_\infty}\right| \label{eq: absolute val} \\
& = \frac{|\varepsilon_0|}{\|u\|_\infty} \label{eq: use vareps}\\
 & = \frac{O_P(N^{1/2})}{\|u\|_\infty} \label{eq: vareps is small} \\
& = O_P(N^{-1/2}) \label{eq: consistency final}
\end{align}
where the \eqref{eq: absolute val} holds because $(|a| - |b|)^2 \leq (a-b)^2$; \eqref{eq: use vareps} holds by definition of $\varepsilon_0$ in \eqref{eq: vareps}; \eqref{eq: vareps is small} holds as $|\varepsilon_0| = O_P(N^{-1/2})$ by the arguments given above; and \eqref{eq: consistency final} holds since \eqref{eq: consistency condition} implies $\frac{1}{\|u\|_\infty} \leq 1/C_{\min}$. This proves the theorem.
\end{proof}

\section{Proof of Theorem \ref{th: minimax}} \label{sec: proof minimax}

We will use $f$ to denote the collection of potential outcome mappings $(f_1, \ldots, f_N)$, so that we may write 
\[ Y = f(X) \]
Additionally, we will use $\theta^*(X;f) \equiv \theta^*$ to denote that $\theta^*$ is determined by $X$ and $f$. Using this notation, the target value $L$ may be written as
\begin{equation} \label{eq: proof minimax L}
 L = \frac{|u^T(f(X) - \theta^*(X;f)|}{\|u\|_\infty}.
\end{equation}
Let $\textup{Err}^*$ denote the minimax error, 
\[ \textup{Err}^* = \min_{\hat{L} \in \mathcal{L}} \max_{\{f_i\}_{i=1}^N \in \mathcal{F}} \mathbb{E}[(\hat{L} - L)^2]\]
which may be written using \eqref{eq: proof minimax L} as 
\[ \textup{Err}^* = \min_{\hat{L}} \quad \max_{f \in \mathcal{F}} \quad \mathbb{E}\left[ \left( \frac{|u^T(f(X) - \theta^*(X;f))|}{\|u\|_\infty} - \hat{L}(X)\right)^2\right] \]
where $\mathcal{F}$ denotes the space of possible outcome mappings $f$.

To lower bound the minimax error, we will construct mappings $f^1$ and $f^2$ in $\mathcal{F}$ as follows. First, let $f^* \in \mathcal{F}$ satisfy
\begin{align} \label{eq: f star}
f^*(X) & \in \argmax_{y \in \{0,1\}^N} |u^Ty|, & X \in \{0,1\}^N 
\end{align}
so that if $Y = f^*(X)$, then $Y$ maximizes $|u^TY|$. Let $\mathcal{E}$ denote the event that at least 2 units are treated,
\[ \mathcal{E} = \{X: \sum_{i=1}^N X_i \geq 2\} \]
which holds with probability 1 under Assumption \ref{as: random sampling}. Let $f^k$ for $k=1,2$ satisfy
\begin{align} \label{eq: f1 f2 in E}
f^k(X) & = f^*(X) \textup{ for all } X \in \mathcal{E} 
\end{align}
so that $f^1$ and $f^2$ return identical outcomes determined by $f^*$ whenever at least 2 units are treated. To describe the outcomes when $X \notin \mathcal{E}$ and only 0 or 1 units are treated, let $\theta^1$ and $\theta^2$ abbreviate $\theta^*(X; f^1)$ and $\theta^*(X, f^2)$ respectively, and let them satisfy
\begin{align} \label{eq: f1 f2 theta}
\theta^2 &= 1 - \theta^1
\end{align}
We note that \eqref{eq: f1 f2 in E} describes the outcomes under $f^1$ and $f^2$ for $X \in \mathcal{E}$, and hence places no constraints on $\theta^1$ and $\theta^2$, which describe the outcome under $f^1$ and $f^2$ under isolated treatment, in which case $X \notin \mathcal{E}$. As a result, the mapping $\theta^*(\ \cdot\ , f^1)$ can be any member of the set $\Theta$ given by 
\begin{equation}\label{eq: Theta set}
\Theta = \{ \theta^*(\ \cdot\ ,f): f \in \mathcal{F}\}
\end{equation}

Our proof will require the following identity, which is proven in the supplemental materials.
\begin{lemma} \label{le: minimax proof 1}
Given a random variable $x$ taking values in $\mathcal{X}$, let $\mathcal{H}$ denote the space of functions $\mathcal{X}\mapsto \mathbb{R}$. Given functions $g_1, g_2 \in \mathcal{H}$, it holds that
\begin{align}
\nonumber & \min_{h \in \mathcal{H}} \max\bigg\{ \ \mathbb{E}[(g_1(x) - h(x))^2],\  \mathbb{E}[(g_2(x) - h(x))^2] \ \bigg\} \\
 & \hskip.5cm {} = \mathbb{E}\left[\left(\frac{g_1(x) - g_2(x)}{2}\right)^2\right] \label{eq: minimax proof identity 1}
\end{align}
\end{lemma}

\noindent We now prove Theorem \ref{th: minimax}.

\begin{proof}[Proof of Theorem \ref{th: minimax}] 

To prove the theorem, we first upper bound the error of the proposed estimator, then lower bound the minimax error, and then show that the bounds are equal.

As the proposed estimator $\hat{L}$ given by \eqref{eq: hat L} is equal to
\[ \hat{L} = \frac{|u^T(Y - 0.5)|}{\|u\|_\infty},\]
and the target value $L$ is given by
\[ L = \frac{|u^T(Y - \theta^*)|}{\|u\|_\infty},\]
it follows that the error satisfies
\begin{align*}
(\hat{L} - L)^2 &\leq \left(\frac{u^T(\theta^* - 0.5)}{\|u\|_\infty}\right)^2,
\end{align*}
where we have used the identity that $(|a| - |b|)^2 \leq (a - b)^2$, letting $a = u^T(Y-0.5)$ and $b = u^T(Y - \theta^*)$. Taking expectations yields the upper bound 
\begin{align}\label{eq: minimax proof 1}
\mathbb{E}[(\hat{L} - L)^2] \leq \mathbb{E}\left[ \left(\frac{u^T(\theta^* - 0.5)}{\|u\|_\infty}\right)^2\right]
\end{align}


We now lower bound $\textup{Err}^*$, the minimax error. To lighten notation we suppress $X$ in functions, writing $f$, $f^*$, and $\theta^*(f)$ respectively for $f(X)$,  $f^*(X)$, and $\theta^*(X;f)$. Additionally we use $\hat{L}$ to abbreviate $\hat{L}(X,Y)$. The minimax error can be lower bounded as follows:
\begin{align}
\label{eq: lower bound 1} \textup{Err}^* &\geq \min_{\hat{L} \in \mathcal{L}} \ \ \max_{f \in \{f^1,f^2\}} \mathbb{E}\left[ \left( \frac{|u^T(f - \theta^*(f))|}{\|u\|_\infty} - \hat{L}\right)^2\right]  \\
\label{eq: lower bound 2} &\geq \min_{\hat{L} \in \mathcal{L}} \ \ \max_{f \in \{f^1,f^2\}}  \mathbb{E}\left[ \left( \frac{|u^T(f - \theta^*(f))|}{\|u\|_\infty} - \hat{L}\right)^2 \cdot 1\{X \in \mathcal{E}\}\right] \\
\label{eq: lower bound 3} & = \min_{\hat{L} \in \mathcal{L}} \ \ \max_{k=1,2} \mathbb{E}\left[ \left( \frac{|u^T(f^* - \theta^*(f^k))|}{\|u\|_\infty} - \hat{L}\right)^2 \cdot 1\{X \in \mathcal{E}\}\right]  \\
\nonumber & = \min_{\hat{L} \in \mathcal{L}} \ \ \max \bigg\{ \mathbb{E}\left[ \left( \frac{|u^T(f^* - \theta^1)|}{\|u\|_\infty} - \hat{L}\right)^2 \cdot 1\{X \in \mathcal{E}\}\right],   \\
\label{eq: lower bound 4} & \hskip2cm {} \mathbb{E}\left[ \left( \frac{|u^T(f^* - (1 - \theta^1))|}{\|u\|_\infty} - \hat{L}\right)^2 \cdot 1\{X \in \mathcal{E}\}\right] \bigg\} \\
\label{eq: lower bound 5} & = \mathbb{E}\left[ \left( \frac{|u^T(f^* - \theta^1)|}{\|u\|_\infty}  - \frac{|u^T(f^* - (1 - \theta^1))|}{\|u\|_\infty} \right)^2 \cdot 1\{X \in \mathcal{E}\} \right] \\
\label{eq: lower bound 6} & = \mathbb{E}\left[ \left(\frac{|u^T(\theta^1 - 0.5)|}{\|u\|_\infty}\right)^2\cdot 1\{X \in \mathcal{E}\} \right] \\
\label{eq: lower bound 8}& = \mathbb{E}\left[ \left(\frac{|u^T(\theta^1 - 0.5)|}{\|u\|_\infty}\right)^2 \right] 
\end{align}
where \eqref{eq: lower bound 1} holds as $\{f^1,f^2\}$ is a subset of $\mathcal{F}$; \eqref{eq: lower bound 2} holds since we have multiplied the squared error by the indicator function, which cannot increase it; \eqref{eq: lower bound 3} holds since $f^1$ and $f^2$ both equal $f^*$ for $X \in \mathcal{E}$ by \eqref{eq: f1 f2 in E}; \eqref{eq: lower bound 4} holds by definition of $\theta^1$ and by \eqref{eq: f1 f2 theta}; \eqref{eq: lower bound 5} holds by Lemma \ref{le: minimax proof 1}; \eqref{eq: lower bound 6} holds because \eqref{eq: f star} implies that
\[ \operatorname{sign}(u^Tf^*) = \operatorname{sign}(u^Tf^* - u^T\theta^1) = \operatorname{sign}(u^Tf^* - u^T(1-\theta^1)),\]
so that
\[ |u^T(f^* - \theta^1)| - |u^T(f^* - (1- \theta^1))| = |u^T(1 - 2\theta^1)|,\]
and \eqref{eq: lower bound 8} holds as $P(X \notin \mathcal{E}) = 0$ under Assumption \ref{as: random sampling}.

We now compare the upper and lower bounds. \eqref{eq: minimax proof 1} upper bounds the estimation error of the proposed $\hat{L}$ by
\[ \mathbb{E}[(\hat{L} - L)^2] \leq \mathbb{E}\left[ \left(\frac{u^T(\theta^*(X; f) - 0.5)}{\|u\|_\infty}\right)^2\right] \]
where we have used $\theta^* = \theta^*(X; f)$. To bound the worst-case performance, we maximize both sides over all $f \in \mathcal{F}$, yielding
\begin{align} \label{eq: minimax proof end}
\max_{f \in \mathcal{F}} \, \mathbb{E}[(\hat{L} - L)^2] \leq \max_{f \in \mathcal{F}}\ \mathbb{E}\left[ \left(\frac{u^T(\theta^*(X; f) - 0.5)}{\|u\|_\infty}\right)^2\right]. 
\end{align}
On the other hand, \eqref{eq: lower bound 8} lower bounds the minimax error by
\[ \textup{Err}^* \geq \mathbb{E}\left[ \left(\frac{|u^T(\theta^*(X; f^1) - 0.5)|}{\|u\|_\infty}\right)^2 \right]  \]
where we have used $\theta^1 = \theta^*(X; f^1)$. Since the mapping  $\theta^*(\cdot, f^1)$ can be any member of the set $\Theta$ given by \eqref{eq: Theta set}, it follows that
\[ \textup{Err}^* \geq \max_{f \in \mathcal{F}} \mathbb{E}\left[ \left(\frac{|u^T(\theta^*(X;f) - 0.5)|}{\|u\|_\infty}\right)^2 \right]  \]
which combines with \eqref{eq: minimax proof end} to prove the claim.

\end{proof}




 

\section{Proof of Lemmas \ref{le: V consistency} and \ref{le: minimax proof 1}} \label{sec: proof minimax helper}

Here we prove Lemmas \ref{le: V consistency}, \ref{le: minimax proof 1}, and their required helper lemmas.

\subsection{Proof of Lemma \ref{le: V consistency}} \label{sec: V consistency}

Our proof of Lemma \ref{le: V consistency} will use the following tools. The Azuma-Hoeffding (or McDiarmid's) inequality \cite[Thm 6.2]{boucheron2013concentration} states that given independent variables $U = (U_1,\ldots,U_N)$, and a function satisfying the bounded difference property
\begin{equation} \label{eq: mcdiarmid condition}
 |f(U) - f(U')| \leq c_i \text{ if } U_j = U_j' \text{ for all } j \neq i,
\end{equation}
it holds that 
\begin{equation}\label{eq: mcdiarmid}
\mathbb{P}\left( | f(U) - \mathbb{E}f | > \epsilon \right) \leq 2\exp\left( - 2\epsilon^2/\sum_i c_i^2 \right).
\end{equation} 
To apply Azuma-Heoffding, we will require the following lemma which is proven in Section \ref{sec: proof Qij}.

\begin{lemma}\label{le: Qij}
Let the conditions of Theorem \ref{th: coverage} hold, and let $E_2$ denote the set of unit pairs $(i,j)$ such that $i$ and $j$ are either connected or have a common neighbor in $G$.  Then there exists constants $C_4$ and $N_0$ such that for all $(i,j) \notin E_2$, $N \geq N_0$, and all $X_i, X_j \in \{0,1\}$, it holds that
\begin{equation}\label{eq: le Qij claim}
\Big|  \mathbb{E}[u_iu_j|X_i,X_j] \Big| \leq C_4/N
\end{equation}

\end{lemma}

We now prove Lemma \ref{le: V consistency}.

\begin{proof}[Proof of Lemma \ref{le: V consistency}]
We first show that $\E[V] = \Var \tau$:
\begin{align}
\nonumber \E[V] &= \sum_{i=1}^N \sum_{j=1}^N \E\Big[\theta_i^* \theta_j^* \E[u_i u_j|X_i, X_j]\Big] \\
\label{eq: expectation} & = \sum_{i=1}^N \sum_{j=1}^N \E\Big[ \E[ \theta_i^* \theta_j^* u_i u_j|X_i, X_j] \Big] \\
\label{eq: iterated expectation} & = \sum_{i=1}^N \sum_{j=1}^N \E\Big[ \theta_i^* \theta_j^* u_i u_j \Big] \\
\nonumber & =  \E\Big[ \Big(u^T \theta^* \Big)^2 \Big] \\
\label{eq: var is ex2} & = \operatorname{Var} \tau,
\end{align}
where \eqref{eq: expectation} holds because $\theta_i^*$ and $\theta_j^*$ are functions of $X_i$ and $X_j$, as $h(X_i, X_j)\cdot \mathbb{E}[g(X)|X_i, X_j] = \E[h(X_i, X_j)\cdot g(X)|X_i, X_j]$ for any functions $g, h$; \eqref{eq: iterated expectation} is the law of iterated expectations; and \eqref{eq: var is ex2} holds since $\tau = u^T\theta^*$ and $\E\tau = 0$.

We next show that $V$ converges to its expectation. Let $U = (U_1,\ldots,U_N)$ denote a vector of $N$ i.i.d.\! Uniform[0,1] random variates. Then the treatment vector $X$ may be generated by treating the $T$ units with the smallest values in the vector $U$, so that
\[ X_i = \begin{cases} 1 & \textup{ if } U_i \leq U_{[T]} \\ 0 & \textup{ else, }\end{cases}\qquad i \in [N], \]
where $U_{[T]}$ denotes the $T$th empirical quantile of $\{U_i\}_{i=1}^N.$ 

As $V$ is generated by $U$, a collection of $N$ independent random variables, we may apply Azuma-Hoeffding. To do so, we write $V$ as 
\[ V = \sum_{i=1}^N \sum_{j=1}^N \bar{Q}_{ij}(X_i,X_j)\]
for $\bar{Q}$ given by
\[ \bar{Q}_{ij}(X_i,X_j) = \theta_i^* \theta_j^* \mathbb{E}[u_i u_j| X_i,X_j]\] 
It can be seen that changing a single entry of $U$ affects at most two entries of the vector $X$, as at most 1 unit goes from untreated to treated and vice versa, and that each entry of $X$ affects 1 row and 1 column of the matrix $\bar{Q}$. To bound the contribution of each row of $\bar{Q}$ to $V$, we observe that 
\begin{align}
\nonumber \sum_{j=1}^N |\bar{Q}_{ij}| &= \sum_{j: (i,j) \in E_2} |\bar{Q}_{ij}| + \sum_{j: (i,j) \notin E_2} |\bar{Q}_{ij}| \\
\label{eq: apply le Qij} & \leq \sum_{j: (i,j) \in E_2} C^2 + \sum_{j: (i,j) \notin E_2} C_4/N \\
\label{eq: count entries} & \leq (d_{\max} + d_{\max}^2)C^2 + C_4
\end{align}
where \eqref{eq: apply le Qij} applies $|u_i| \leq C$ (from Assumption \ref{as: tau}) and $\theta_i^* \in \{0,1\}$ to bound $|\bar{Q}_{ij}|$ for $(i,j) \in E_2$, and applies Lemma \ref{le: Qij} to bound $|\bar{Q}_{ij}| \leq C_4/N$ for $(i,j) \notin E_2$; and \eqref{eq: count entries} uses that each unit has at most $d_{\max}$ neighbors and $d_{\max}^2$ units with whom it shares a common neighbor to bound  $|\{j: (i,j) \in E_2\}|$, and also uses the trivial bound $|\{j: (i,j) \notin E_2\}| \leq N$.

Using \eqref{eq: count entries}, we may apply Azuma Hoeffding with $c_i = (d_{\max} + d_{\max}^2)C^2 + C_4$ to show that
\[ P\left( |V - \E V| \geq \epsilon\right) \leq \exp\left( -\frac{2\epsilon^2}{N((d_{\max} + d_{\max}^2)C^2 + C_4)}\right)\]
which implies that $|V - \E V| = O_P(\sqrt{N}) = o_P(\operatorname{Var} \tau)$, where the last equality holds since the lemma assumes $\operatorname{Var} \tau > N^{1/2 + \delta}$ for $\delta > 0$. As $\E V = \operatorname{Var} \tau$ by \eqref{eq: var is ex2}, this proves the claim of the lemma.
\end{proof}

\subsection{Proof of Lemma \ref{le: minimax proof 1}}

\begin{proof}[Proof of Lemma \ref{le: minimax proof 1}]
We first show that given random variables $a$ and $b$, it holds that
\begin{align}
\max\Big\{\, \E[(a+b)^2], \, \E[(a-b)^2]\, \Big\} \geq \E[a^2] \label{eq: minimax helper 0}
\end{align}
To show \eqref{eq: minimax helper 0}, we observe that
\begin{align*}
& \max\Big\{\, \E[(a+b)^2], \, \E[(a-b)^2]\, \Big\} \\
& = \max\Big\{\, \E[a^2 + 2ab + b^2], \, \E[a^2 - 2ab + b^2]\, \Big\} \\
& = \E[a^2 + b^2] + \max\Big\{\, \E[2ab], \, -\E[2ab]\, \Big\} \\
& = \E[a^2] + \E[b^2] + \Big|\E[2ab]\Big| \\
& \geq \E[a^2]
\end{align*}
where the final inequality holds as $\E[b^2]$ and $|\E[2ab]|$ are both nonnegative.

We now show the result. Let $h(X) = (g_1(X) + g_2(X))/2 + R(X)$ and let $\delta(X) = (g_1(X) - g_2(X))/2$. It can then be shown that
\begin{align}
\nonumber &  \max\bigg\{ \ \mathbb{E}[(g_1(x) - h(x))^2],\  \mathbb{E}[(g_2(x) - h(x))^2] \ \bigg\} \\
\nonumber & = \max\bigg\{\ \mathbb{E}[(\delta(X) + R(X))^2], \  \mathbb{E}[(\delta(X) - R(X))^2]\ \bigg\} \\
& \geq \mathbb{E}[\delta(X)^2] \label{eq: minimax helper}
\end{align}
where \eqref{eq: minimax helper} follows from \eqref{eq: minimax helper 0}. As \eqref{eq: minimax helper} holds with equality if $R(X)=0$, the claim of the lemma follows.
\end{proof}

\subsection{Proof of Lemma \ref{le: Qij}} \label{sec: proof Qij}

Recall that $\eta_i$ denotes the neighborhood of unit $i$, as defined by \eqref{eq: eta gamma}. Given $X$ and $i \in [N]$, let $X_{\eta_i}$ denote the subvector of $X$ corresponding to the units in $\eta_i$. The proof of Lemma \ref{le: Qij} will require Lemma \ref{le: Qij helper}, which relates the joint and marginal distributions of $X_{\eta_i}$ and $X_{\eta_j}$ conditional on $X_i$ and $X_j$, when $\{X_i, \eta_i\}$ and $\{X_j, \eta_j\}$ are disjoint. It is proven in Section \ref{sec: proof Qij helper}.

\begin{lemma}\label{le: Qij helper}
Let the conditions of Lemma \ref{le: Qij} hold. Then there exists $C_5$ and $N_0$ such that if $N \geq N_0$ it holds for all $(i,j) \notin E_2$, all $X_i, X_j \in \{0,1\}$ and all $x_{\eta_i} \in \{0,1\}^{|\eta_i|}$ and $x_{\eta_j} \in \{0,1\}^{|\eta_j|}$ that

\begin{equation}
\left|\frac{P(X_{\eta_i} = x_{\eta_i}, X_{\eta_j} = x_{\eta_j}| X_i, X_j)}{P(X_{\eta_i} = x_{\eta_i}| X_i) P(X_{\eta_j} = x_{\eta_j}| X_j)} - 1 \right| \leq \frac{C_5}{N}
\end{equation}
\end{lemma}

\begin{proof}[Proof of Lemma \ref{le: Qij}]

As $u_i(X)$ depends on $X$ only through $X_i$ and $X_{\eta_i}$, given $x \in \{0,1\}$ and $g \in \{0,1\}^{|\eta_i|}$, let $u_i(x, g)$ denote the value of $u_i(X)$ when $X_i =x$ and $X_{\eta_i} = g$. 

Given $g_1\in \{0,1\}^{|\eta_i|}$ and $g_2 \in \{0,1\}$, let $r(g_1,g_2)$ satisfy
\begin{equation}\label{eq: r}
\frac{P(X_{\eta_i} = g_1, X_{\eta_j} = g_2| X_i, X_j)}{P(X_{\eta_i} = x_{\eta_i}| X_i) P(X_{\eta_j} = x_{\eta_j}| X_j)}  = 1 + r(g_1,g_2)
\end{equation}
For the $N_0$ described in Lemma \ref{le: Qij helper}, it holds for all $N \geq N_0$ that
\begin{align}
\nonumber & \Big| \mathbb{E}[u_i u_j|X_i, X_j] \Big| \\
 & = \left| \sum_{g_1, g_2} P(X_{\eta_i} = g_1, X_{\eta_j} = g_2|X_i, X_j) u_i(X_i,g_1)u_j(X_j,g_2)\right| \label{eq: Qij definition of E}\\
 & = \left| \sum_{g_1, g_2} P(X_{\eta_i} = g_1|X_i) P(X_{\eta_j} = g_2|X_j) u_i(X_i,g_1)u_j(X_j,g_2) (1 + r(g_1, g_2)) \right| \label{eq: Qij definition of r} \\
 & = \left| \sum_{g_1, g_2} P(X_{\eta_i} = g_1|X_i) P(X_{\eta_j} = g_2|X_j) u_i(X_i,g_1)u_j(X_j,g_2) \cdot r(g_1, g_2) \right| \label{eq: Qij u is mean zero}\\
 & \leq  \frac{C^2 C_5}{N} \label{eq: Qij u r are bounded},
\end{align}
Here \eqref{eq: Qij definition of E} follows by definition of the conditional expectation; \eqref{eq: Qij definition of r} follows by \eqref{eq: r}; \eqref{eq: Qij u is mean zero} follows as 
\begin{align*}
& \sum_{g_1, g_2} P(X_{\eta_i} = g_1|X_i) P(X_{\eta_j} = g_2|X_j) u_i(X_i,g_1)u_j(X_j,g_2) \\
& = \mathbb{E}[u_i|X_i] \mathbb{E}[u_j|X_j] \\
& = 0,
\end{align*}
since $\mathbb{E}[u_i|X_i] = 0$ for all $i \in [N]$ by Assumption \ref{as: tau}; and \eqref{eq: Qij u r are bounded} holds for $N \geq N_0$ because $|u_i| \leq C$ by Assumption \ref{as: tau clt} and $|r(g_1,g_2)| \geq C_5/N$ by Lemma \ref{le: Qij helper}. Letting $C_4 = C^2C_5$ proves the lemma.
\end{proof}

\subsection{Proof of Lemma \ref{le: Qij helper}} \label{sec: proof Qij helper}

\begin{proof}[Proof of Lemma \ref{le: Qij helper}]

Given units $(i,j)$, and given vectors $x_1 \in \{0,1\}^{|\eta_i|}$ and $x_2 \in \{0,1\}^{|\eta_j|}$, let $d_1 = |\eta_i|$ and $d_2 = |\eta_j|$ denote the respective dimensions of $x_1$ and $x_2$, and let $o_1$ and $o_2$ denote the sum of the entries in $x_1$ and $x_2$ respectively,
\begin{align*}
o_1 & = \sum_{k=1}^{d_1} [x_1]_k  & o_2 & = \sum_{k=1}^{d_2} [x_2]_k
\end{align*}
Given nonnegative integers $a,b$ such that $a \geq b$, let $P^a_b$ denote $\frac{a!}{(a-b)!}$. For $a,a',b$ such that $P^a_b$ and $P^{a'}_b$ are defined, we will use the identity that if  
\begin{align}\label{eq: P bound condition}
\left|\frac{a-a'}{a'-b+1}\right| \leq 1
\end{align}
 then it holds that\footnote{To see \eqref{eq: P bound}, observe that 
 \begin{align}
\frac{P^a_b}{P^{a'}_b} & = \frac{a(a-1)\cdots(a-b+1)}{a'(a'-1)\cdots(a'-b+1)} 
 = \prod_{k=0}^{b-1}\left(1 + \frac{a-a'}{a'-k}\right) 
 = \sum_{s \in \{0,1\}^b} \prod_{k=1}^b \left(\frac{a-a'}{a'-k+1}\right)^{s_k} \label{eq: binomial expansion} 
\end{align} 
Under \eqref{eq: P bound condition}, for all $s \neq 0$ it holds that 
\begin{align}
\left|\prod_{k=1}^b \left(\frac{a-a'}{a'-k+1}\right)^{s_k}\right| \leq  \left| \frac{a-a'}{a'-b+1}\right|, \label{eq: P condition implication}
\end{align}
and using \eqref{eq: P condition implication} to bound the terms appearing in the final equation of \eqref{eq: binomial expansion} implies \eqref{eq: P bound}.}
\begin{align} \label{eq: P bound}
\left| \frac{P^a_b}{P^{a'}_b} - 1\right| \leq (2^b-1) \left| \frac{a-a'}{a'-b+1}\right|.
\end{align}
To lighten notation, let $\tilde{a}$ and $\tilde{b}$ be defined as
\begin{align} \label{eq: a b definitions}
\tilde{a} & = N - 2 & \tilde{b} & = T - X_i - X_j
\end{align}
and given $x_1 \in \{0,1\}^{|\eta_i|}$ and $x_2 \in \{0,1\}^{|\eta_j|}$, let $\{\epsilon_i\}_{i=1}^6$ define terms that will appear in a later equation \eqref{eq: epsilons}, 
\begin{align}
\epsilon_1 & = \frac{P^{\tilde{a}-\tilde{b}}_{d_1 - o_1}}{P^{\tilde{a}+1-(\tilde{b}+X_j)}_{d_1-o_1}}-1, & 
\epsilon_2 & = \frac{P^{\tilde{a}-\tilde{b} - (d_1-o_1)}_{d_2 - o_2}}{P^{\tilde{a}+1-(\tilde{b}+X_i)}_{d_2-o_2}}-1, & 
\epsilon_3 & = \frac{P^{\tilde{a}+1}_{d_1}}{P^{\tilde{a}}_{d_1}} - 1 \label{eq: eps def 1}\\
\epsilon_4 & = \frac{P^{\tilde{a}+1}_{d_2}}{P^{\tilde{a}-d_1}_{d_2}} - 1, &
\epsilon_5 & = \frac{P^{\tilde{b}}_{o_1}}{P^{\tilde{b}+X_j}_{o_1}} -1,  &
\epsilon_6 & = \frac{P^{\tilde{b} - o_1}_{o_2}}{P^{\tilde{b}+X_i}_{o_2}} -1 \label{eq: eps def 2}
\end{align}
As $\tilde{a}$, $\tilde{b}$, and $\tilde{a} - \tilde{b}$ are all affine functions of $N$ under Assumption \ref{as: random sampling}, while $d_1, d_2, o_1$, and $o_2$ are bounded by constants under Assumption \ref{as: degree}, for large enough $N$ we may apply \eqref{eq: P bound} to the  definitions of $\epsilon_i$ given by \eqref{eq: eps def 1}-\eqref{eq: eps def 2}. For example, applying \eqref{eq: P bound} to $\epsilon_1$ as defined in \eqref{eq: eps def 1} results in the following bound for $|\epsilon_1|$,
\begin{align*}
|\epsilon_1| & \leq (2^{d_1-o_1}-1)\left|\frac{1-X_j}{\tilde{a}-\tilde{b}+1-X_j-(d_1-o_1)+1}\right| 
\end{align*}
which holds for $N$ large enough that
\[ \left|\frac{1 - X_j}{\tilde{a} - \tilde{b} + 1 - X_j - (d_1 - o_1) + 1}\right| \leq 1, \]
with $|\epsilon_2|$, $|\epsilon_3|$, $|\epsilon_4|$, $|\epsilon_5|$, and $|\epsilon_6|$ having similar bounds. In turn, these bounds imply for some $C_4', N_0'$ and all $N \geq N_0'$ that 
\begin{align} \label{eq: eps bound final}
 |\epsilon_i| &\leq \frac{C_4'}{N} & i=1,\ldots,6.
\end{align}
We now prove the result. For $(i,j) \notin E_2$ and any $x_1 \in \{0,1\}^{|\eta_i|}$ and $x_2 \in \{0,1\}^{|\eta_j|}$ it can be seen that
\begin{align}
\nonumber & \frac{P(X_{\eta_i} = x_1, X_{\eta_j} = x_2|X_i,X_j)}{P(X_{\eta_i} = x_1|X_i)P(X_{\eta_j} = x_2|X_j)}\\
\label{eq: counting arguments} & = \frac{{{N-2-d_1-d_2} \choose {T - X_i - X_j -o_1-o_2}}}{{{N-2} \choose {T - X_i - X_j}}} 
\cdot \frac{{{N-1} \choose {T - X_i}}}{{{N-1-d_1} \choose {T - X_i - o_1}}} 
\cdot \frac{{{N-1} \choose {T - X_j}}}{{{N-1-d_2} \choose {T - X_j - o_2}}} \\
\label{eq: a b} & = \frac{{{\tilde{a}-d_1-d_2} \choose {\tilde{b}-o_1-o_2}}}{{{\tilde{a}} \choose {\tilde{b}}}} 
\cdot \frac{{{\tilde{a}+1} \choose {\tilde{b} + X_j}}}{{{\tilde{a}+1-d_1} \choose {\tilde{b} + X_j - o_1}}} 
\cdot \frac{{{\tilde{a}+1} \choose {\tilde{b} + X_i}}}{{{\tilde{a}+1-d_2} \choose {\tilde{b} + X_i - o_2}}} \\
\label{eq: P identity} & = \frac{P^{\tilde{a}-\tilde{b}}_{d_1+d_2-o_1-o_2}  P^{\tilde{b}}_{o_1 + o_2}}{P^{\tilde{a}}_{d_1+d_2}} 
\cdot \frac{P^{\tilde{a}+1}_{d_1}}{P^{\tilde{a}+1-(\tilde{b}+X_j)}_{d_1-o_1}  P^{\tilde{b}+X_j}_{o_1}} 
\cdot \frac{P^{\tilde{a}+1}_{d_2}}{P^{\tilde{a}+1-(\tilde{b}+X_i)}_{d_2-o_2}  P^{\tilde{b}+X_i}_{o_2}} \\
\label{eq: reorder P} & = \frac{P^{\tilde{a}-\tilde{b}}_{d_1+d_2-o_1-o_2}}{P^{\tilde{a}+1-(\tilde{b}+X_j)}_{d_1-o_1}P^{\tilde{a}+1-(\tilde{b}+X_i)}_{d_2-o_2}}
\cdot\frac{P^{\tilde{b}}_{o_1 + o_2}}{P^{\tilde{b}+X_j}_{o_1}P^{\tilde{b}+X_i}_{o_2}}
\cdot\frac{P^{\tilde{a}+1}_{d_1}P^{\tilde{a}+1}_{d_2}}{P^{\tilde{a}}_{d_1+d_2}} 
\\
\label{eq: P split} & = \frac{P^{\tilde{a}-\tilde{b}}_{d_1 - o_1}}{P^{\tilde{a}+1-(\tilde{b}+X_j)}_{d_1-o_1} } 
\cdot \frac{P^{\tilde{a}-\tilde{b}-(d_1-o_1)}_{d_2-o_2}}{P^{\tilde{a}+1-(\tilde{b}+X_i)}_{d_2-o_2}} 
\cdot \frac{P^{\tilde{a}+1}_{d_1}}{P^{\tilde{a}}_{d_1}} 
\cdot \frac{P^{\tilde{a}+1}_{d_2}}{P^{\tilde{a}-d_1}_{d_2}}
\cdot \frac{P^{\tilde{b}}_{o_1}}{P^{\tilde{b}+X_j}_{o_1}} 
\cdot \frac{P^{\tilde{b} - o_1}_{o_2}}{P^{\tilde{b}+X_i}_{o_2}} \\
\label{eq: epsilons} & = (1+\epsilon_1)(1+\epsilon_2)(1+\epsilon_3)(1+\epsilon_4)(1+\epsilon_5)(1+\epsilon_6)
\end{align}
where \eqref{eq: counting arguments} holds by counting arguments\footnote{For example, 
\[ P(X_{\eta_i}=x_1,|X_i) = \frac{{{N-1 - d_i} \choose {T-X_i - o_i}}}{{{N-1} \choose {T-X_i}}},\]
where the numerator counts the number of treatment assignments that assign $x_1$ to $X_{\eta_i}$ and $X_i$ its observed value, effectively leaving $T-X_i-o_i$ units to be selected for treatment  (out of $N - 1 - d_i$ units total), and the denominator similarly counts the number of treatments assignments that assign $X_i$ and $X_j$ their observed values.}; \eqref{eq: a b} substitutes the definitions of $\tilde{a}$ and $\tilde{b}$ given by \eqref{eq: a b definitions}; \eqref{eq: P identity} uses the identity 
\[ \frac{{a-c \choose b-d}}{{{a} \choose {b}}} = \frac{ P^{a-b}_{c-d}\, \cdot\, P^b_d}{P^a_c}\] 
while \eqref{eq: reorder P} reorders terms; \eqref{eq: P split} uses the identity
$P^{a}_{b_1 + b_1} = P^{a}_{b_1} \cdot P^{a-b_1}_{b_2}$; and \eqref{eq: epsilons} holds by definition of $\{\epsilon_i\}_{i=1}^6$ in \eqref{eq: eps def 1}-\eqref{eq: eps def 2}.

Combining \eqref{eq: eps bound final} with \eqref{eq: epsilons} implies the lemma.
\end{proof}

\pagebreak 

 \bibliographystyle{apalike}
 {\footnotesize
 \bibliography{bibfile}

\begin{thebibliography}{}

\bibitem[Aronow, 2012]{aronow2012general}
Aronow, P.~M. (2012).
\newblock A general method for detecting interference between units in
  randomized experiments.
\newblock {\em Sociological Methods \& Research}, 41(1):3--16.

\bibitem[Aronow and Samii, 2017]{aronow2017estimating}
Aronow, P.~M. and Samii, C. (2017).
\newblock Estimating average causal effects under general interference, with
  application to a social network experiment.
\newblock {\em The Annals of Applied Statistics}, 11(4):1912--1947.

\bibitem[Athey et~al., 2018]{athey2017exact}
Athey, S., Eckles, D., and Imbens, G.~W. (2018).
\newblock Exact p-values for network interference.
\newblock {\em Journal of the American Statistical Association},
  113(521):230--240.

\bibitem[Basse et~al., 2019a]{basse2019randomization_peer}
Basse, G., Ding, P., Feller, A., and Toulis, P. (2019a).
\newblock Randomization tests for peer effects in group formation experiments.
\newblock {\em arXiv preprint arXiv:1904.02308}.

\bibitem[Basse et~al., 2019b]{basse2019randomization}
Basse, G., Feller, A., and Toulis, P. (2019b).
\newblock Randomization tests of causal effects under interference.
\newblock {\em Biometrika}, 106(2):487--494.

\bibitem[Billingsley, 2013]{billingsley2013convergence}
Billingsley, P. (2013).
\newblock {\em Convergence of probability measures}.
\newblock John Wiley \& Sons.

\bibitem[Boucheron et~al., 2013]{boucheron2013concentration}
Boucheron, S., Lugosi, G., and Massart, P. (2013).
\newblock {\em Concentration inequalities: A nonasymptotic theory of
  independence}.
\newblock Oxford university press.

\bibitem[Chatterjee, 2008]{chatterjee2008new}
Chatterjee, S. (2008).
\newblock A new method of normal approximation.
\newblock {\em Ann. Probab.}, 36(1):1584--1610.

\bibitem[Chen et~al., 2004]{chen2004normal}
Chen, L.~H., Shao, Q.-M., et~al. (2004).
\newblock Normal approximation under local dependence.
\newblock {\em The Annals of Probability}, 32(3):1985--2028.

\bibitem[Chin, 2018]{chin2018central}
Chin, A. (2018).
\newblock Central limit theorems via stein's method for randomized experiments
  under interference.
\newblock {\em arXiv preprint arXiv:1804.03105}.

\bibitem[Choi, 2021]{choi2021randomization}
Choi, D. (2021).
\newblock New estimands for experiments with strong interference.
\newblock {\em arXiv preprint arXiv:2107.00248}.

\bibitem[Cortez et~al., 2022]{cortez2022staggered}
Cortez, M., Eichhorn, M., and Yu, C. (2022).
\newblock Staggered rollout designs enable causal inference under interference
  without network knowledge.
\newblock In {\em Advances in Neural Information Processing Systems}.

\bibitem[Dubhashi and Panconesi, 2009]{dubhashi2009concentration}
Dubhashi, D.~P. and Panconesi, A. (2009).
\newblock {\em Concentration of measure for the analysis of randomized
  algorithms}.
\newblock Cambridge University Press.

\bibitem[Forastiere et~al., 2021]{forastiere2020identification}
Forastiere, L., Airoldi, E.~M., and Mealli, F. (2021).
\newblock Identification and estimation of treatment and interference effects
  in observational studies on networks.
\newblock {\em Journal of the American Statistical Association},
  116(534):901--918.

\bibitem[Han et~al., 2022]{han2022detecting}
Han, K., Li, S., Mao, J., and Wu, H. (2022).
\newblock Detecting interference in a/b testing with increasing allocation.
\newblock {\em arXiv preprint arXiv:2211.03262}.

\bibitem[Hu et~al., 2021]{hu2021average}
Hu, Y., Li, S., and Wager, S. (2021).
\newblock Average direct and indirect causal effects under interference.
\newblock {\em Biometrika}.

\bibitem[Leung, 2022]{leung2022rate}
Leung, M.~P. (2022).
\newblock Rate-optimal cluster-randomized designs for spatial interference.
\newblock {\em The Annals of Statistics}, 50(5):3064--3087.

\bibitem[Li and Wager, 2022a]{li2022network}
Li, S. and Wager, S. (2022a).
\newblock Network interference in micro-randomized trials.
\newblock {\em arXiv preprint arXiv:2202.05356}.

\bibitem[Li and Wager, 2022b]{li2022random}
Li, S. and Wager, S. (2022b).
\newblock Random graph asymptotics for treatment effect estimation under
  network interference.
\newblock {\em The Annals of Statistics}, 50(4):2334--2358.

\bibitem[Li et~al., 2019]{li2019randomization}
Li, X., Ding, P., Lin, Q., Yang, D., and Liu, J.~S. (2019).
\newblock Randomization inference for peer effects.
\newblock {\em Journal of the American Statistical Association}, pages 1--31.

\bibitem[Miguel and Kremer, 2004]{miguel2004worms}
Miguel, E. and Kremer, M. (2004).
\newblock Worms: identifying impacts on education and health in the presence of
  treatment externalities.
\newblock {\em Econometrica}, pages 159--217.

\bibitem[Ogburn et~al., 2022]{ogburn2022causal}
Ogburn, E.~L., Sofrygin, O., Diaz, I., and Van~der Laan, M.~J. (2022).
\newblock Causal inference for social network data.
\newblock {\em Journal of the American Statistical Association}, pages 1--15.

\bibitem[Park and Kang, 2023]{park2023assumption}
Park, C. and Kang, H. (2023).
\newblock Assumption-lean analysis of cluster randomized trials in infectious
  diseases for intent-to-treat effects and network effects.
\newblock {\em Journal of the American Statistical Association},
  118(542):1195--1206.

\bibitem[Pouget-Abadie et~al., 2019]{pouget2019testing}
Pouget-Abadie, J., Saint-Jacques, G., Saveski, M., Duan, W., Ghosh, S., Xu, Y.,
  and Airoldi, E.~M. (2019).
\newblock Testing for arbitrary interference on experimentation platforms.
\newblock {\em Biometrika}, 106(4):929--940.

\bibitem[Rinott, 1994]{rinott1994normal}
Rinott, Y. (1994).
\newblock On normal approximation rates for certain sums of dependent random
  variables.
\newblock {\em Journal of Computational and Applied Mathematics},
  55(2):135--143.

\bibitem[S{\"a}vje et~al., 2021]{savje2017average}
S{\"a}vje, F., Aronow, P.~M., and Hudgens, M.~G. (2021).
\newblock {Average treatment effects in the presence of unknown interference}.
\newblock {\em The Annals of Statistics}, 49(2):673 -- 701.

\bibitem[Tchetgen and VanderWeele, 2012]{tchetgen2012causal}
Tchetgen, E. J.~T. and VanderWeele, T.~J. (2012).
\newblock On causal inference in the presence of interference.
\newblock {\em Statistical methods in medical research}, 21(1):55--75.

\bibitem[Tchetgen~Tchetgen et~al., 2021]{tchetgen2021auto}
Tchetgen~Tchetgen, E.~J., Fulcher, I.~R., and Shpitser, I. (2021).
\newblock Auto-g-computation of causal effects on a network.
\newblock {\em Journal of the American Statistical Association},
  116(534):833--844.

\bibitem[Wang, 2021]{wang2021causal}
Wang, Y. (2021).
\newblock Causal inference under temporal and spatial interference.
\newblock {\em arXiv preprint arXiv:2106.15074}.

\end{thebibliography}
 }


\end{document}